\documentclass[aps,pra,onecolumn,notitlepage,superscriptaddress,nofootinbib,longbibliography]{revtex4-2}

\usepackage{amsmath,amssymb,amsthm}
\usepackage{graphicx}
\usepackage{hyperref}
\usepackage{bm}
\usepackage{booktabs}
\usepackage{tikz}
\usetikzlibrary{arrows.meta,calc,positioning,patterns,decorations.pathmorphing,matrix}

\newtheorem{theorem}{Theorem}
\newtheorem{lemma}{Lemma}
\newtheorem{proposition}{Proposition}
\newtheorem{corollary}{Corollary}
\theoremstyle{definition}
\newtheorem{definition}{Definition}
\newtheorem{example}{Example}
\newtheorem{remark}{Remark}

\newcommand{\F}{\mathbb{F}}
\newcommand{\Z}{\mathbb{Z}}
\newcommand{\rank}{\mathrm{rank}}
\newcommand{\wt}{\mathrm{wt}}
\newcommand{\Span}{\mathrm{span}}
\newcommand{\row}{\mathrm{row}}
\newcommand{\Ker}{\mathrm{ker}}

\newcommand{\Dir}{\mathcal{D}} % direction alphabet
\newcommand{\word}[1]{\texttt{#1}} % allows embedded math mode (e.g. NE$^2$NE$^2$N)

\DeclareMathOperator{\Ann}{Ann}        % annihilator (ideal/module annihilator)
\begin{document}

\title{Structural Analysis of Directional qLDPC Codes}

\author{Mohammad Rowshan}
\email{m.rowshan@unsw.edu.au}
\affiliation{University of New South Wales (UNSW), Kensington, NSW 2052, Australia}

\date{\today}

\begin{abstract}
Directional codes, recently introduced by Gehér--Byfield--Ruban \cite{Geher2025Directional}, constitute a hardware-motivated family of quantum low-density parity-check (qLDPC) codes. These codes are defined by stabilizers measured by ancilla qubits executing a fixed \emph{direction word} (route) on square- or hex-grid connectivity. In this work, we develop a comprehensive \emph{word-first} analysis framework for route-generated, translation-invariant CSS codes on rectangular tori. Under this framework, a direction word $W$ deterministically induces a finite support pattern $P(W)$, from which we analytically derive: 
(i)~a closed-form route-to-support map; 
(ii)~the odd-multiplicity difference lattice $L(W)$ that classifies commutation-compatible $X/Z$ layouts; and 
(iii)~conservative finite-torus admissibility criteria. 
Furthermore, we provide: 
(iv)~a rigorous word equivalence and canonicalization theory (incorporating dihedral lattice symmetries, reversal/inversion, and cyclic shifts) to enable symmetry-quotiented searches; 
(v)~an ``inverse problem'' criterion to determine when a translation-invariant support pattern is realizable by a single route, including reconstruction and non-realizability certificates; and 
(vi)~a quasi-cyclic (group-algebra) reduction for row-periodic layouts that explains the sensitivity of code dimension $k$ to boundary conditions. 
As a case study, we analyze the word $W=\word{NE$^2$NE$^2$N}$ end-to-end. We provide explicit stabilizer dependencies, commuting-operator motifs, and an exact criterion for dimension collapse on thin rectangles: for $(L_x, L_y) = (2d, d)$ with row alternation, we find $k=4$ if $6 \mid d$, and $k=0$ otherwise.
\end{abstract}

\maketitle

%\tableofcontents

% ==========================================================
\section{Introduction}
% ==========================================================

Quantum error correction (QEC) enables reliable quantum computation by encoding logical information into many physical qubits.
The stabilizer and CSS frameworks convert QEC into linear-algebraic constraints over $\F_2$ \cite{Shor1995Scheme,Calderbank1996Good,Steane1996Error,Gottesman1997Stabilizer}.
In classical coding theory, low-density parity-check (LDPC) codes are defined by sparse parity-check graphs and admit efficient iterative decoding \cite{Gallager1962Lowdensity,Tanner1981Recursive,MacKay1999Good}.
Quantum LDPC (qLDPC) codes adapt this sparsity to the stabilizer setting \cite{MacKay2004Sparse,Terhal2015Review,Errorcorrectionzoo2025QLDPC}.

A practical issue is that ``good'' asymptotic qLDPC parameters do not automatically imply implementable layouts.
Two-dimensional topological codes such as toric and surface codes have strictly local checks and mature decoders \cite{Kitaev2003Fault,Dennis2002Topological,Fowler2012Surface}, but 2D locality imposes strong scaling constraints \cite{Bravyi2010Tradeoffs,Bravyi2009Nogotheorem}.
Meanwhile, many powerful qLDPC constructions (products, Cayley/expander ideas, asymptotically good families) are not natively aligned with strict 2D connectivity \cite{Tillich2014Quantum,Couvreur2013Construction,Bravyi2014Homological,Leverrier2015Quantumexpander,Panteleev2021Asymptotically,Panteleev2022Almostlinear,Hastings2021Fiber,Breuckmann2021Balanced}.
Recent progress shows that algebraic structure can be paired with engineering constraints: bivariate bicycle codes demonstrate high thresholds and low overhead \cite{Bravyi2024Highthreshold} and motivate architectural proposals \cite{Yoder2025Tour}, while connectivity-reduction techniques such as morphing circuits \cite{Shaw2025Lowering} and hardware-aware placement/routing studies \cite{Mathews2025Placing,Berthusen2025Toward} further emphasize co-design.

Directional codes, introduced in \cite{Geher2025Directional}, take the co-design philosophy to an extreme: the stabilizer support is generated by a short \emph{direction word} (a route) that every ancilla follows on a low-degree lattice.
This route-generated restriction makes strict square/hex connectivity natural, but it also introduces delicate structural dependence on the word and on boundary conditions:
small changes in $W$ can alter (i) which $X/Z$ layouts are allowed, (ii) which finite tori admit collision-free schedules, and (iii) which boundary conditions preserve or collapse the dimension~$k$.

\subsection{Relationship to the seminal work and %originality of 
contributions}
The construction and core layout-classification mechanism are due to Geh{\'e}r--Byfield--Ruban \cite{Geher2025Directional}.
Sections~\ref{sec:construction}--\ref{sec:invariants} represent these foundational ingredients from \cite{Geher2025Directional} in a word-centric language with explicit examples and diagrams (route-to-support formula; odd-multiplicity lattice governing commutation-compatible layouts; conservative finite-torus admissibility filters). 
The contributions of this work begin in Sections~\ref{sec:equiv}--\ref{sec:qc}:
\begin{enumerate}
\item \emph{Word equivalence and canonical representatives} (Sec.~\ref{sec:equiv}): we formalize an equivalence group generated by dihedral lattice symmetries, reversal with inversion, and cyclic shifts, prove that it induces code isomorphisms on compatible tori, and give a canonicalization procedure for symmetry-quotiented searches.
\item \emph{The route realizability inverse problem} (Sec.~\ref{sec:inverse}): given a translation-invariant offset set $P$, when does there exist a single direction word $W$ such that $P=P(W)$? We give a reconstruction criterion and practical non-realizability certificates.
\item \emph{A quasi-cyclic (group-algebra) reduction for row-periodic layouts} (Sec.~\ref{sec:qc}): for common layouts such as row alternation, we reduce stabilizer-dependency and rank questions to annihilator computations in a small polynomial ring.
As an application we prove an exact divisibility criterion explaining the observed $k$-collapse for $W=\word{NE$^2$NE$^2$N}$ on thin rectangles (Theorem~\ref{thm:kc-collapse}).
\end{enumerate}

\subsection{Outline}
Section~\ref{sec:construction} develops the checkerboard model and the route-to-support map. %, with embedded examples and diagrams.
Section~\ref{sec:invariants} introduces the odd-multiplicity difference lattice and its layout classification theorem. %, again with examples and diagrams.
Section~\ref{sec:equiv} develops word symmetries and canonicalization.
Section~\ref{sec:inverse} treats the inverse problem of route realizability.
Section~\ref{sec:finitetorus} discusses finite-torus admissibility and introduces the polynomial/QC viewpoint, while Section~\ref{sec:qc} develops the row-periodic annihilator formalism.
Section~\ref{sec:families} provides worked families and a detailed case study.
Section~\ref{sec:numerics} documents the reproducible numerical pipeline and representative scans \cite{rowshan2026dircoode}.
Finally, Section~\ref{sec:conclusion} summarizes and discusses future directions.

% ==========================================================
\section{Checkerboard Model and Route-Generated Checks (after \cite{Geher2025Directional})}\label{sec:construction}
% ==========================================================

We adopt the standard CSS setting \cite{Calderbank1996Good,Steane1996Error,Gottesman1997Stabilizer}, but present it in the coordinate system that makes directional codes convenient to compute with.

\subsection{CSS codes and parameters}
A CSS code is specified by two binary matrices $H_X,H_Z$ with $H_XH_Z^\top=0$.
The number of logical qubits is
\begin{equation}
k=n-\rank(H_X)-\rank(H_Z),
\end{equation}
and the (static) distances are
\begin{align}
d_X &= \min\{\wt(x): x\in\Ker(H_Z)\setminus \row(H_X)\},\\
d_Z &= \min\{\wt(z): z\in\Ker(H_X)\setminus \row(H_Z)\},
\end{align}
with $d=\min(d_X,d_Z)$.

\subsection{Checkerboard torus and layouts}
% We work on the 2D torus
% \[
% G=\Z_{L_x}\times \Z_{L_y},
% \]
% with $L_x,L_y$ even, so the checkerboard split is balanced. Coordinates are taken modulo $(L_x,L_y)$.
% We split sites by parity
% \begin{align}
% \Lambda_Q &= \{(x,y)\in G : x+y\equiv 0\pmod 2\} \quad\text{(data)},\\
% \Lambda_A &= \{(x,y)\in G : x+y\equiv 1\pmod 2\} \quad\text{(ancilla)}.
% \end{align}
% Thus, $|\Lambda_Q|=|\Lambda_A|=n=\tfrac12L_xL_y$.
% A \emph{layout} is a partition $\Lambda_A=\Lambda_X\sqcup\Lambda_Z$; equivalently, a labeling function $\alpha:\Lambda_A\to\{X,Z\}$.

We work in 2D grid $G=\Z_{L_x}\times \Z_{L_y}$ with $L_x,L_y$ even, to make the checkerboard split perfectly balanced. Here, $\Z_{L}$ denotes integers modulo of $L$ (that is, $0,1,\dots,L-1$ with wraparound).  Since the grid ``wraps'' in both directions $x,y$, it is a torus topology. 
We split sites by parity
\begin{align}
\Lambda_Q &= \{(x,y)\in G : x+y\equiv 0\pmod 2\} \quad\text{(data)},\\
\Lambda_A &= \{(x,y)\in G : x+y\equiv 1\pmod 2\} \quad\text{(ancilla)}.
\end{align}
Thus, $|\Lambda_Q|=|\Lambda_A|=n=\tfrac12L_xL_y$.
A \emph{layout} is a partition $\Lambda_A=\Lambda_X\sqcup\Lambda_Z$; every ancilla site will be assigned to measure either an X-check (goes in $\Lambda_X$), or a Z-check (goes in $\Lambda_Z$). Equivalently, a label function
$\alpha:\Lambda_A\to\{X,Z\}$ tells you the check type for each ancilla site. 
Figure~\ref{fig:checkerboard} illustrates an example layout for $G=\Z_8\times \Z_6$. 
On a checkerboard lattice, data qubits live on one sublattice, and ancillas live on the other sublattice.

\begin{figure}[t]
\centering
\begin{tikzpicture}[scale=0.55]
  \draw[step=1,gray!25,very thin] (0,0) grid (8,6);
  \foreach \x in {0,...,7}{
    \foreach \y in {0,...,5}{
      \pgfmathtruncatemacro{\p}{mod(\x+\y,2)}
      \ifnum\p=0
        \fill[black] (\x+0.5,\y+0.5) circle (2.0pt);
      \fi
    }
  }
  % example layout: row alternation
  \foreach \x in {0,...,7}{
    \foreach \y in {0,...,5}{
      \pgfmathtruncatemacro{\p}{mod(\x+\y,2)}
      \ifnum\p=1
        \pgfmathtruncatemacro{\py}{mod(\y,2)}
        \ifnum\py=0
          \node at (\x+0.5,\y+0.5) {\scriptsize X};
        \else
          \node at (\x+0.5,\y+0.5) {\scriptsize Z};
        \fi
      \fi
    }
  }
\end{tikzpicture}
\caption{Checkerboard data/ancilla partition (shown as a patch) and a simple example layout with row alternation.
Here, black dots are data qubits ($x+y$ even) and X/Z labels mark ancillas ($x+y$ odd).}
\label{fig:checkerboard}
\end{figure}

\subsection{Direction words and compressed notation}
A \emph{direction word} is a string $W=d_1d_2\cdots d_w$ with letters from $\Dir=\{N,E,S,W\}$, interpreted as steps
\[
N=(0,1),\qquad E=(1,0),\qquad S=(0,-1),\qquad W=(-1,0).
\]
We also use a compact notation: for example $\word{NE$^2$N}$ means $\word{N E E N}$, and $\word{NE$^2$NE$^2$N}$ means $\word{N E E N E E N}$.

\subsection{From a word to a support pattern}
Let the partial sums  $S_j=\sum_{t=1}^j d_t$ with $S_0=(0,0)$ be the running ``walk location'', which gives the position after $j$ steps, if you start at the origin. For example, after step 1, you are at $S_1=d_1$, after step 2 you are at $S_2=d_1+d_2$, etc. 
%In directional codes, the data qubit involved at step $j$ depends only on the ancilla position before and after that step.
%Because these data qubits lie at midpoints of consecutive route positions, we use a doubled integer coordinate convention, leading to the following closed form \cite{Geher2025Directional}.

In directional codes, the order in which an ancilla interacts with data qubits follows a fixed route on the lattice. At each step, the data qubit involved depends only on the ancilla’s position before and after that step. Because these data qubits lie halfway between consecutive route positions, we multiply coordinates by two so that all positions are integers. This leads to a simple formula that directly gives the data-qubit location for each step. %\cite{Geher2025Directional}.

\begin{lemma}[Route-to-offset formula]\label{lem:offset}
Define integer offsets
\begin{equation}
Q_j=S_{j-1}+S_j = 2\sum_{t=1}^{j-1} d_t + d_j,\qquad j=1,\dots,w.
\end{equation}
Then an ancilla anchored at $a\in\Lambda_A$ interacts with data qubits at sites $a+Q_j$ (mod $G$).
\end{lemma}

\begin{remark}[Route interpretation as a local-gate schedule]\label{rem:route-not-motion}
Although a direction word is drawn as a geometric route, no qubit is physically transported across the lattice.
Instead, the word specifies a \emph{sequence of local two-qubit interactions} that propagates the ancilla's \emph{quantum state} along nearest-neighbor edges.
Data qubits remain fixed, and each interaction is local in the underlying connectivity graph.
Directional codes thus trade strict geometric locality of a single stabilizer (as in plaquette codes) for a fixed, bounded-depth measurement schedule compatible with low-degree hardware graphs \cite{Geher2025Directional}.
\end{remark}

\begin{definition}[Support pattern]
The support pattern of $W$ is the set
\begin{equation}
P(W)=\{Q_1,\dots,Q_w\}\subset\Z^2,
\end{equation}
understood modulo $(L_x,L_y)$ when embedded in $G$.
A directional check anchored at $a$ acts on data qubits in the translated support $a+P(W)$.
\end{definition}
Note that the shifted (translated) locations of the offset set $P(W)$ by the anchor position $a$ are the data qubits included in the stabilizer check. Thus, $W$ gives the shape and $a$ gives the placement.

The construction becomes conceptually simple:
\emph{choosing $W$ chooses $P(W)$}, hence fixes the check shape (up to translation), and the remaining degrees of freedom are (i) the layout (which ancillas are $X$ or $Z$) and (ii) the boundary conditions (the torus vectors).

\paragraph{Pattern-code viewpoint.}
Once $P(W)$ is fixed, directional codes fit inside translation-invariant \emph{pattern CSS codes}:
\begin{equation}
S_a^X=\prod_{q\in a+P(W)} X_q,\quad a\in\Lambda_X,
\end{equation}
\begin{equation}
S_b^Z=\prod_{q\in b+P(W)} Z_q,\quad b\in\Lambda_Z,
\end{equation}
where $\Lambda_X,\Lambda_Z$ form a partition of the ancillas (the \emph{layout}).
%The core question becomes: \emph{for a given $W$, which layouts are allowed?}
%The answer is encoded by a lattice computed from $P(W)$ (Section~\ref{sec:invariants}).

%%

\begin{remark}
    We cover the whole grid by translating the same route-defined check to every ancilla position. The union of all translated supports $\bigcup_{a \in \Lambda_A}(a+P(W))$ spreads across the entire data lattice, so every data qubit participates in checks.
\end{remark}

\begin{example}[word $\rightarrow$ check \emph{and} a row of $H$]
\begin{equation*}
P(W)=\{(0,1),(1,2),(3,2),(4,3)\},
\end{equation*}
where the offsets are obtained as $Q_1=S_0+S_1=(0,0)+(0,1)=(0,1)$, $Q_2=S_1+S_2=(0,1)+(1,1)=(1,2)$, etc. 
For the anchor $a=(1,0)\in\Lambda_A$, the translated support is
\begin{equation*}
a+P(W)=\{(1,1),(2,2),(4,2),(5,3)\}\subset\Lambda_Q,
\end{equation*}
where $x+y$ is even parity for all positions. %If we fix any ordering of data qubits, this support becomes a length-$n$ binary row with ones in the four corresponding columns.
If we index data qubits in lexicographic order by looping over rows increasing $y=0,1,\dots,5$, and inside each row, looping over columns in increasing $x=0,1,,\dots,7$ 
(\emph{skipping} ancilla sites), consideing only sites where $x+y$ is even, these coordinates correspond to column indices
$\{4,9,10,14\}$ of $H$. Therefore, the parity-check row contributed by this ancilla is
\begin{equation*}
h_a = e_4+e_9+e_{10}+e_{14}\in \mathbb{F}_2^{n}.
\end{equation*}
Figure~\ref{fig:toycheck-row} shows this translation geometrically. 
To get the entire $H$ matrix, Loop over all ancillas $a \in \Lambda_A$ and place that row in $H_X$ if $\alpha(a)=X$ or $H_Z$ if $\alpha(a)=Z$. Stack all rows and we get $\left|\Lambda_A\right|=24$ rows where each row has weight $\boldsymbol{w}=4$ for \word{NE$^2$N}. 
%Whether this row is placed into $H_X$ or $H_Z$ is determined by the layout function $\alpha(a)$. 

\end{example}

\begin{figure}[t]
\centering
\begin{tikzpicture}[scale=0.55]
  % draw a small patch of the checkerboard
  \foreach \x in {0,...,7} {
    \foreach \y in {0,...,5} {
      \pgfmathtruncatemacro{\p}{mod(\x+\y,2)}
      \ifnum\p=0
        \fill[black] (\x,\y) circle (1.6pt); % data
      \else
        \fill[gray]  (\x,\y) circle (1.6pt); % ancilla
      \fi
    }
  }

  % anchor ancilla a=(1,0)
  \draw[thick] (1,0) circle (3.0pt);
  \node[below] at (1,0) {\scriptsize $a$};

  % support a+P(W) for W=NE2N : (1,1),(2,2),(4,2),(5,3)
  \foreach \x/\y/\lab in {1/1/4, 2/2/9, 4/2/10, 5/3/14} {
    \draw[thick] (\x,\y) circle (3.0pt);
    \node[above right] at (\x,\y) {\scriptsize $\lab$};
  }

  % light guide edges (optional, helps reading)
  \draw[thin,gray!40] (1,0)--(1,1);
  \draw[thin,gray!40] (1,1)--(2,2);
  \draw[thin,gray!40] (2,2)--(4,2);
  \draw[thin,gray!40] (4,2)--(5,3);
\end{tikzpicture}
\caption{Toy instance illustrating the terminology. Gray sites are ancillas, black sites are data.
For $W=\word{NE$^2$N}$ and anchor $a=(1,0)$ the four circled data qubits are the translated support $a+P(W)$.
%The small numbers show the corresponding data-column indices for one concrete ordering convention,
%so the row inserted into $H_X$ or $H_Z$ is $e_4+e_9+e_{10}+e_{14}$.
}
\label{fig:toycheck-row}
\end{figure}
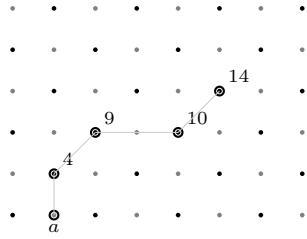

\subsection{Representative support patterns}
Figure~\ref{fig:fivewords} shows $P(W)$ for five representative words.
The first four are benchmark examples in \cite{Geher2025Directional}; the fifth is our case study.

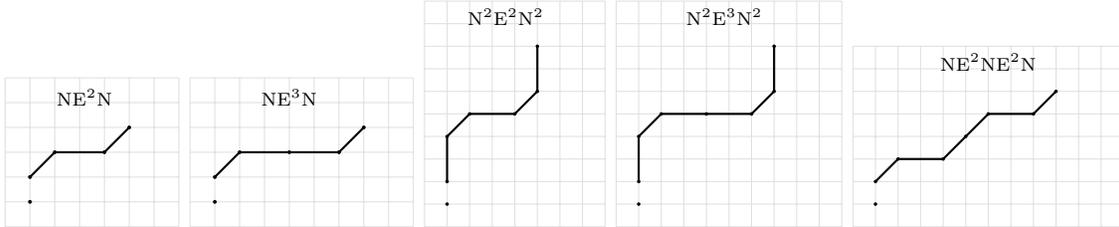
\begin{figure*}[t]
\centering
\begin{tabular}{ccccc}
\begin{tikzpicture}[scale=0.33]
  \draw[step=1,gray!25,very thin] (-1,-1) grid (6,5);
  \fill (0,0) circle (2.2pt);
  \foreach \x/\y in {0/1,1/2,3/2,4/3}{
    \fill (\x,\y) circle (2.2pt);
  }
  \draw[thick] (0,1)--(1,2)--(3,2)--(4,3);
  \node at (2.2,4.2) {\scriptsize NE$^2$N};
\end{tikzpicture}
&
\begin{tikzpicture}[scale=0.33]
  \draw[step=1,gray!25,very thin] (-1,-1) grid (8,5);
  \fill (0,0) circle (2.2pt);
  \foreach \x/\y in {0/1,1/2,3/2,5/2,6/3}{
    \fill (\x,\y) circle (2.2pt);
  }
  \draw[thick] (0,1)--(1,2)--(3,2)--(5,2)--(6,3);
  \node at (3.0,4.2) {\scriptsize NE$^3$N};
\end{tikzpicture}
&
\begin{tikzpicture}[scale=0.30]
  \draw[step=1,gray!25,very thin] (-1,-1) grid (7,9);
  \fill (0,0) circle (2.2pt);
  \foreach \x/\y in {0/1,0/3,1/4,3/4,4/5,4/7}{
    \fill (\x,\y) circle (2.2pt);
  }
  \draw[thick] (0,1)--(0,3)--(1,4)--(3,4)--(4,5)--(4,7);
  \node at (2.6,8.3) {\scriptsize N$^2$E$^2$N$^2$};
\end{tikzpicture}
&
\begin{tikzpicture}[scale=0.30]
  \draw[step=1,gray!25,very thin] (-1,-1) grid (9,9);
  \fill (0,0) circle (2.2pt);
  \foreach \x/\y in {0/1,0/3,1/4,3/4,5/4,6/5,6/7}{
    \fill (\x,\y) circle (2.2pt);
  }
  \draw[thick] (0,1)--(0,3)--(1,4)--(3,4)--(5,4)--(6,5)--(6,7);
  \node at (3.8,8.3) {\scriptsize N$^2$E$^3$N$^2$};
\end{tikzpicture}
&
\begin{tikzpicture}[scale=0.30]
  \draw[step=1,gray!25,very thin] (-1,-1) grid (11,7);
  \fill (0,0) circle (2.2pt);
  \foreach \x/\y in {0/1,1/2,3/2,4/3,5/4,7/4,8/5}{
    \fill (\x,\y) circle (2.2pt);
  }
  \draw[thick] (0,1)--(1,2)--(3,2)--(4,3)--(5,4)--(7,4)--(8,5);
  \node at (5.0,6.3) {\scriptsize NE$^2$NE$^2$N};
\end{tikzpicture}
\end{tabular}
\caption{Support patterns $P(W)$ for five representative words, plotted as integer offsets relative to an ancilla anchor at the origin.}
\label{fig:fivewords}
\end{figure*}

% ==========================================================
\section{Layout Classification via Odd-Multiplicity Differences (after \cite{Geher2025Directional})}\label{sec:invariants}
% ==========================================================

Directional codes impose a global constraint: not every layout $\alpha$ is compatible with a given $W$.
The mechanism is overlap parity: an $X$-check and a $Z$-check anticommute iff they overlap on an odd number of data qubits.
Since all checks share the same shape $P(W)$ up to translation, overlap parity depends only on the displacement between anchors.
This section turns that intuition into a computable invariant. %\cite{Geher2025Directional}.

\subsection{Odd-multiplicity differences and the forced-equality lattice}
Let $P(W)=\{Q_1,\dots,Q_w\}$.
Define the difference multiset
\begin{equation}
\Delta(W)=\{v=Q_j-Q_i: 1\le i<j\le w\},
\end{equation}
with multiplicity $\mu(v)$ counting how many pairs yield the same $v$.

\begin{definition}[Odd-multiplicity difference set and lattice]
Define
\begin{equation}
\Delta_{\mathrm{odd}}(W)=\{v\in\Z^2:\mu(v)\equiv 1 \ (\mathrm{mod}\ 2)\},
\end{equation}
and the \emph{odd-multiplicity difference lattice}
\begin{equation}
L(W)=\Span_{\Z}\big(\Delta_{\mathrm{odd}}(W)\big).
\end{equation}
\end{definition}

\begin{remark}[Intuition for $L(W)$]
Suppose a displacement $\delta$ is ``forbidden'' (i.e., it produces an odd overlap). Then ancillas at
$a$ and $a+\delta$ are not allowed to have opposite types, so they must share the same label.
Applying the same rule repeatedly shows that ancillas at
$a,\,a+\delta,\,a+2\delta,\,a+3\delta,\dots$ must all have the same label; in other words, labels are
constant along all integer multiples of $\delta$.
If there are several forbidden displacements, the same chaining works for any sum of them.
Therefore we collect all integer combinations of forbidden displacements into the lattice $L(W)=\Span_{\mathbb{Z}}\big(\Delta_{\mathrm{odd}}(W)\big),$ 
and $L(W)$ is exactly the set of displacements along which ancilla labels are forced to be equal.
\end{remark}

%\subsection{Geometric illustration of an odd-overlap displacement}
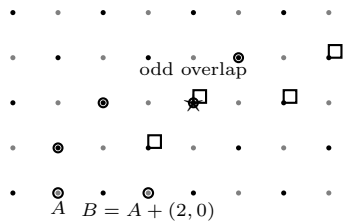
\begin{figure}[t]
\centering
\begin{tikzpicture}[scale=0.60]
  \foreach \x in {0,...,7} {
    \foreach \y in {0,...,4} {
      \pgfmathtruncatemacro{\p}{mod(\x+\y,2)}
      \ifnum\p=0
        \fill[black] (\x,\y) circle (1.6pt);
      \else
        \fill[gray]  (\x,\y) circle (1.6pt);
      \fi
    }
  }

  % two ancilla anchors A and B=A+(2,0)
  \draw[thick] (1,0) circle (3.0pt);
  \node[below] at (1,0) {\scriptsize $A$};
  \draw[thick] (3,0) circle (3.0pt);
  \node[below] at (3,0) {\scriptsize $B=A+(2,0)$};

  % supports for W=NE2N: A targets (circles)
  \foreach \x/\y in {1/1,2/2,4/2,5/3} {
    \draw[thick] (\x,\y) circle (2.8pt);
  }

  % B targets (squares)
  \foreach \x/\y in {3/1,4/2,6/2,7/3} {
    \draw[thick] (\x,\y) rectangle ++(0.28,0.28);
  }

  % overlap at (4,2)
  \node at (4,2) {\Large $\star$};
  \node[above] at (4,2.35) {\scriptsize odd overlap};

\end{tikzpicture}
\caption{For $W=\word{NE$^2$N}$, anchors separated by $\delta=(2,0)$ have supports that overlap on exactly one data qubit (star).
Thus $\delta=(2,0)$ forbids assigning opposite types to those ancillas. This is the concrete meaning of an ``odd-difference'' displacement.}
\label{fig:toycheck-overlap}
\end{figure}

\subsection{Worked example: computing $\Delta_{\mathrm{odd}}$ and $L(W)$}
\begin{example}[computing $\Delta_{\mathrm{odd}}$ for NE$^2$N]
For $W=\texttt{NE$^2$N}$, we have $P(W)=\{(0,1),(1,2),(3,2),(4,3)\}$. 
Computing all $Q_j-Q_i$ with $j>i$, we get  
$Q_2-Q_1=(1,1)$, $Q_3-Q_1=(3,1)$, 
$Q_4-Q_1=(4,2)$, $Q_3-Q_2=(2,0)$, 
 $Q_4-Q_2=(3,1)$, and $Q_4-Q_3=(1,1)$. 
The differences include $(2,0)$ and $(4,2)$ exactly once each, while $(1,1)$ and $(3,1)$ occur twice.
Thus, the forbidden ancilla displacements are 
\[
\Delta_{\mathrm{odd}}(\texttt{NE$^2$N})=\{(2,0),(4,2)\}.
\]

%To obtain their integer span, 
Take integer combinations, including subtractions
$$
m(2,0)+n(4,2)=(2 m+4 n, 2 n),
$$
and rewrite it as
$$
(2(m+2 n), 2 n)=(2k,2\ell),
$$
where $k=m+2n$ and $\ell=n$. Then all vectors have even $x$ and $y$. So, a basis for $L(W)$ can be
\begin{equation*}
L(\texttt{NE$^2$N})=\Span_{\Z}\{(2,0),(4,2)\}=2\Z\times2\Z.
\end{equation*}
%A basis is not unique, and any pair of integer vectors that generates the same lattice is acceptable. 
From the above basis, $(4,2)-2(2,0)=(4,2)-(4,0)=(0,2)$. Physically, if the ancillas at distance $(4,2)$ must agree, and the ancillas at distance $(2,0)$ must agree, then the ancillas at distance $(0,2)$ must also agree. That is transitivity. 
\end{example}

Figure~\ref{fig:toycheck-overlap} geometrically illustrates the odd-overlap displacement $(2,0)$ in this set.

\begin{remark}[Choosing a basis for $L(W)$]
%Given $\Delta_{\mathrm{odd}}(W)$, the lattice $L(W)=\Span_{\mathbb{Z}}\!\big(\Delta_{\mathrm{odd}}(W)\big)$ is obtained by closing the odd-multiplicity displacements under integer linear combinations.
To extract a convenient basis for a given $\Delta_{\mathrm{odd}}(W)$, one may freely add or subtract integer multiples of generators, since such operations do not change the lattice.
In practice, this reduction typically yields a simple rank-two basis consisting of a
``horizontal'' vector $(d,0)$ and a ``diagonal'' or ``vertical'' vector $(e,f)$, or,
after further simplification, purely horizontal and vertical vectors such as $(d,0)$ and $(0,f)$.
The particular choice of basis is not unique; any pair of independent integer vectors generating the same lattice is equally valid.
\end{remark}

\subsection{Layout coset theorem}
\begin{theorem}[Layout coset theorem]\label{thm:layout}
Fix a direction word $W$ on the infinite checkerboard lattice.
Any layout $\alpha:\Lambda_A\to\{X,Z\}$ compatible with directional-code commutation/scheduling constraints must be constant on cosets of $L(W)$:
\begin{equation}
a-b\in L(W)\ \Rightarrow\ \alpha(a)=\alpha(b).
\end{equation}
Conversely, any layout that is constant on cosets of $L(W)$ yields a translation-invariant CSS commutation pattern for checks generated by $P(W)$.
\end{theorem}

\subsection{Coset structure for the two principal lattices}
The two lattice types that arise repeatedly for the benchmark words are illustrated in Fig.~\ref{fig:cosets}.
Layouts correspond to choosing $X/Z$ labels constant on these cosets.

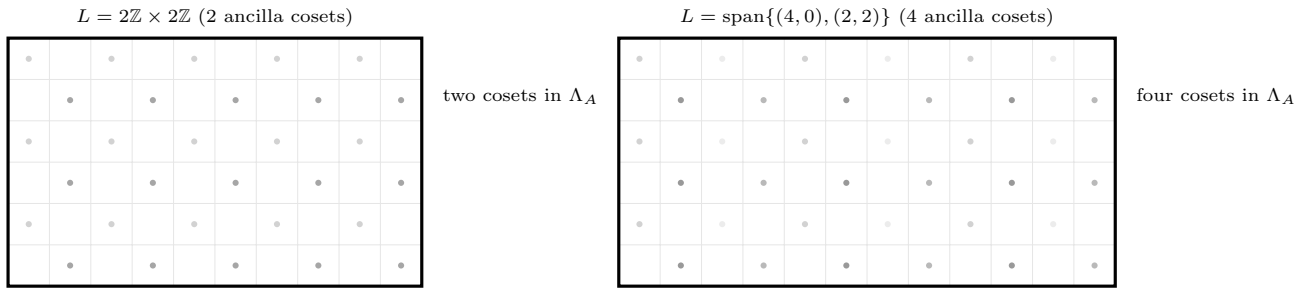
\begin{figure*}[t]
\centering
\begin{tabular}{cc}
\begin{tikzpicture}[scale=0.55]
  \draw[step=1,gray!20,very thin] (0,0) grid (10,6);
  \draw[very thick] (0,0) rectangle (10,6);
  \node at (5,6.5) {\scriptsize $L=2\Z\times2\Z$ (2 ancilla cosets)};
  \foreach \x in {0,...,9}{
    \foreach \y in {0,...,5}{
      \pgfmathtruncatemacro{\p}{mod(\x+\y,2)}
      \ifnum\p=1
        \pgfmathtruncatemacro{\c}{mod(\x,2)}
        \ifnum\c=1
          \fill[gray!70] (\x+0.5,\y+0.5) circle (2.1pt);
        \else
          \fill[gray!35] (\x+0.5,\y+0.5) circle (2.1pt);
        \fi
      \fi
    }
  }
  \node[anchor=west] at (10.3,4.6) {\scriptsize two cosets in $\Lambda_A$};
\end{tikzpicture}
&
\begin{tikzpicture}[scale=0.55]
  \draw[step=1,gray!20,very thin] (0,0) grid (12,6);
  \draw[very thick] (0,0) rectangle (12,6);
  \node at (6,6.5) {\scriptsize $L=\Span\{(4,0),(2,2)\}$ (4 ancilla cosets)};
  \foreach \x in {0,...,11}{
    \foreach \y in {0,...,5}{
      \pgfmathtruncatemacro{\p}{mod(\x+\y,2)}
      \ifnum\p=1
        \pgfmathtruncatemacro{\c}{mod(\x,4)}
        \ifnum\c=1
          \fill[gray!80] (\x+0.5,\y+0.5) circle (2.1pt);
        \fi
        \ifnum\c=3
          \fill[gray!55] (\x+0.5,\y+0.5) circle (2.1pt);
        \fi
        \ifnum\c=0
          \fill[gray!35] (\x+0.5,\y+0.5) circle (2.1pt);
        \fi
        \ifnum\c=2
          \fill[gray!15] (\x+0.5,\y+0.5) circle (2.1pt);
        \fi
      \fi
    }
  }
  \node[anchor=west] at (12.3,4.6) {\scriptsize four cosets in $\Lambda_A$};
\end{tikzpicture}
\end{tabular}
\caption{Coset structure on ancillas for two lattice types that occur frequently in directional-code words.
A layout is any $X/Z$ labeling constant on these cosets (Theorem~\ref{thm:layout}).}
\label{fig:cosets}
\end{figure*}

\subsection{Invariant summary table for representative words}
Table~\ref{tab:lattices2} summarizes $\Delta_{\mathrm{odd}}(W)$ and a convenient basis for $L(W)$ for the five representative words in Fig.~\ref{fig:fivewords}.

\begin{table*}[t]
\centering
\caption{Odd-multiplicity invariants for representative words (computed from $P(W)$).}
\label{tab:lattices2}
\begin{tabular}{@{}llllll@{}}
\toprule
Word $W$ & $w$ & $\Delta_{\mathrm{odd}}(W)$ (generators shown) & basis for $L(W)$ & $[\Z^2\!:\!L(W)]$ & ancilla cosets \\ \midrule
NE$^2$N & 4 & $\{(2,0),(4,2)\}$ & $\{(2,0),(0,2)\}$ & 4 & 2 \\
NE$^3$N & 5 & $\{(4,0),(6,2)\}$ & $\{(4,0),(2,2)\}$ & 8 & 4 \\
N$^2$E$^2$N$^2$ & 6 & $\{(2,0),(4,2),(4,6)\}$ & $\{(2,0),(0,2)\}$ & 4 & 2 \\
N$^2$E$^3$N$^2$ & 7 & $\{(4,0),(6,2),(6,6)\}$ & $\{(4,0),(2,2)\}$ & 8 & 4 \\
NE$^2$NE$^2$N & 7 & $\{(2,2),(6,2),(8,4)\}$ & $\{(4,0),(2,2)\}$ & 8 & 4 \\
\bottomrule
\end{tabular}
\end{table*}

% ==========================================================
\section{Word Symmetries and Canonical Representatives}\label{sec:equiv}
% ==========================================================

When scanning words, it is important to quotient by symmetries that only relabel the lattice.
Let $D_4$ be the dihedral symmetry group of the square, acting on directions by rotations/reflections.
Also consider reversal with direction inversion, and cyclic shifts.

\begin{definition}[Word equivalence]\label{def:equiv}
Two words $W,W'$ are equivalent ($W\sim W'$) if they are related by a composition of:
(i) a dihedral action on $\{N,E,S,W\}$,
(ii) reversal with direction inversion ($N\leftrightarrow S$, $E\leftrightarrow W$), and
(iii) a cyclic shift of letters.
\end{definition}

\begin{proposition}[Equivalence induces code isomorphism]\label{prop:equiv}
If $W\sim W'$, then $P(W')$ is an affine image of $P(W)$ under a lattice automorphism and translation.
On any compatible torus embedding, the corresponding directional codes are isomorphic (qubit relabeling and stabilizer relabeling).
\end{proposition}

\subsection{Canonicalization procedure for symmetry-quotiented searches}
For implementation and numerical searches \cite{rowshan2026dircoode}, it is convenient to map each equivalence class to a single \emph{canonical representative}.
One simple method is: generate all words in the orbit of $W$ under the finite group generated by the operations in Definition~\ref{def:equiv}, then choose the lexicographically smallest expanded-letter string (ties broken deterministically).
A detailed description and complexity bound are given in Appendix~\ref{app:canonical}.

% ==========================================================
\section{Route Realizability and the Inverse Problem}\label{sec:inverse}
% ==========================================================

Directional codes form a strict subclass of translation-invariant pattern CSS codes: not every finite offset set $P$ is realizable as $P(W)$ for a \emph{single} route $W$.
This becomes relevant when comparing directional codes to other quasi-cyclic constructions or when proposing a support pattern first and asking whether it admits a directional extraction schedule.

\subsection{A two-step difference identity}
Let $Q_j=S_{j-1}+S_j$ as in Lemma~\ref{lem:offset}.
Then consecutive offsets satisfy
\begin{equation}
Q_{j+1}-Q_j = d_j+d_{j+1}.
\label{eq:deltaQ}
\end{equation}
Thus, consecutive offsets differ by the sum of two cardinal steps.

Define the finite ``two-step difference alphabet''
\[
\Sigma:=\{u+v: u,v\in\Dir\}=\{(\pm2,0),(0,\pm2),(\pm1,\pm1)\}.
\]

\subsection{Route reconstruction and non-realizability certificates}
\begin{proposition}[Route reconstruction and realizability test]\label{prop:reconstruct}
Let $P\subset\Z^2$ be a multiset of size $w$.
An ordering $Q_1,\dots,Q_w$ of $P$ is realizable by a direction word $W=d_1\cdots d_w$ with $P(W)=P$ (in that order) if and only if:
\begin{enumerate}
\item $Q_1\in\Dir$ (the first offset must be a cardinal step), and
\item for each $j$, $\Delta_j:=Q_{j+1}-Q_j\in\Sigma$, and the recursion
\begin{equation}
d_1 := Q_1,\qquad d_{j+1}:=\Delta_j-d_j
\label{eq:recur}
\end{equation}
produces letters $d_2,\dots,d_w\in\Dir$.
\end{enumerate}
If these conditions hold, the reconstructed word is unique for the chosen ordering.
\end{proposition}

\begin{example}[A simple non-realizable offset set]\label{ex:nonrealizable}
Consider $P=\{(1,2),(3,2),(5,2)\}$.
Every point has odd parity ($x+y$ odd), so parity alone does not forbid it, but $P\cap\Dir=\emptyset$ so no ordering can satisfy $Q_1\in\Dir$.
Hence $P$ cannot be realized by a single route, even though it is a plausible translation-invariant check pattern. Fig.~\ref{fig:nonrealizable} illustrates this example. 
\end{example}

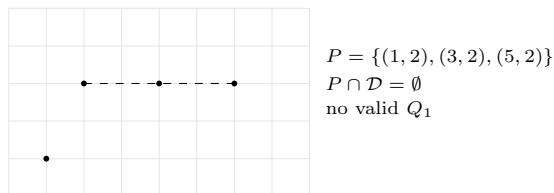
\begin{figure}[t]
\centering
\begin{tikzpicture}[scale=0.50]
  \draw[step=1,gray!20,very thin] (-1,-1) grid (7,4);
  \fill (0,0) circle (2.2pt);
  \foreach \x/\y in {1/2,3/2,5/2}{
    \fill (\x,\y) circle (2.2pt);
  }
  \node[anchor=west] at (7.2,2.7) {\scriptsize $P=\{(1,2),(3,2),(5,2)\}$};
  \node[anchor=west] at (7.2,2.0) {\scriptsize $P\cap\Dir=\emptyset$};
  \node[anchor=west] at (7.2,1.3) {\scriptsize no valid $Q_1$};
  \draw[dashed] (1,2)--(3,2)--(5,2);
\end{tikzpicture}
\caption{A simple non-realizable offset set: it ``looks geometric'' but cannot be ordered to start with a cardinal first offset $Q_1\in\Dir$, as required by Proposition~\ref{prop:reconstruct}.}
\label{fig:nonrealizable}
\end{figure}

% ==========================================================
\section{Finite-Torus Embeddings and Quasi-Cyclic Structure}\label{sec:finitetorus}
% ==========================================================

Directional codes are implemented on finite tori (or finite patches) rather than on $\Z^2$.
Two finite-size effects matter in practice:
(i) wrap-around collisions when the torus is too small compared to $P(W)$, and
(ii) boundary-condition sensitivity of $\rank(H_X),\rank(H_Z)$, hence of $k$.

\subsection{Wrap-around collisions and a conservative admissibility filter}
On a torus, offsets are reduced modulo $(L_x,L_y)$.
If distinct offsets or offset differences collapse modulo the torus, a schedule that is collision-free on the infinite lattice can break. Fig.~\ref{fig:wrap} sketches the issue.

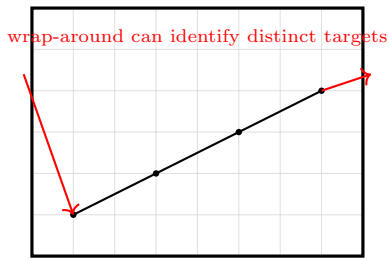
\begin{figure}[t]
\centering
\begin{tikzpicture}[scale=0.55]
  \draw[step=1,gray!25,very thin] (0,0) grid (8,6);
  \draw[very thick] (0,0) rectangle (8,6);
  \fill (1,1) circle (2.2pt);
  \fill (3,2) circle (2.2pt);
  \fill (5,3) circle (2.2pt);
  \fill (7,4) circle (2.2pt);
  \draw[thick] (1,1)--(3,2)--(5,3)--(7,4);
  \draw[thick,->,red] (7,4)--(8.2,4.4);
  \draw[thick,->,red] (-0.2,4.4)--(1,1);
  \node[red] at (4.0,5.3) {\scriptsize wrap-around can identify distinct targets};
\end{tikzpicture}
\caption{Schematic: on small tori, a route may wrap and identify distinct intended targets.}
\label{fig:wrap}
\end{figure}

Let $P(W)=\{Q_1,\dots,Q_w\}$ and define the difference set
\begin{equation}
D(W)=\{\pm(Q_j-Q_i): 1\le i<j\le w\}.
\end{equation}

\begin{proposition}[Rectangle admissibility bound]\label{prop:rectadm}
If
\begin{align}
L_x &> \max_{v\in D(W)\cup (D(W)\pm D(W))} |v_x|,\\
L_y &> \max_{v\in D(W)\cup (D(W)\pm D(W))} |v_y|,
\end{align}
then no nonzero element of $D(W)$ or $D(W)\pm D(W)$ vanishes modulo $(L_x,L_y)$.
In particular, a conservative family of wrap-around collision obstructions is avoided.
\end{proposition}

\subsection{Polynomial encoding of translation-invariant patterns}
Directional codes repeat the same finite support pattern across the lattice by translation, so they are quasi-cyclic in the standard sense.
Given a finite offset pattern $P\subset\Z^2$, define its polynomial
\begin{equation}
h_P(x,y)=\sum_{(a,b)\in P} x^a y^b
\in \F_2[x^{\pm1},y^{\pm1}]/(x^{L_x}-1,y^{L_y}-1).
\end{equation}
Translating $P$ by $(u,v)$ corresponds to multiplying by $x^uy^v$.
Layouts and sublattice restrictions can be encoded by small polynomial matrices whose block-circulant expansion yields $H_X,H_Z$.

% ==========================================================
\section{Quasi-cyclic rank prediction and boundary-condition sensitivity}\label{sec:qc-rank}
% ==========================================================
% Suggested placement:
% Insert this section immediately after your current quasi-cyclic/pattern-polynomial discussion
% (i.e., after Sec.~\ref{sec:finitetorus} or after your existing ``Quasi-cyclic description'' section).
% It upgrades the QC viewpoint into a usable (and provable) $k$-predictor.

Directional codes are translation-invariant in the sense that every check is a translate (shifted copy) of the same finite pattern $P(W)$.
This makes the check matrices block-circulant, so questions about $\rank(H_X),\rank(H_Z)$ (hence $k$) can be reduced to algebra in a small quotient ring.
The main takeaway is that \emph{boundary-condition sensitivity of $k$ is not mysterious}: it corresponds to whether certain low-degree polynomials divide the torus relations, a fact that becomes explicit in the quasi-cyclic (QC) representation.

\subsection{Coarse coordinates and a block-circulant normal form}\label{subsec:coarse-qc}

We keep the checkerboard torus $G=\Z_{L_x}\times\Z_{L_y}$ with $L_x,L_y$ even.
Translations by \emph{even} vectors preserve the data/ancilla bipartition.
Let
\[
G_0:=\Z_{L_x/2}\times \Z_{L_y/2},
\]
which parametrizes even translations via $(i,j)\mapsto (2i,2j)$.
Each parity class in $G$ (e.g.\ data sites with $(x,y)\equiv(0,0)\!\!\pmod 2$) is a torsor for $G_0$.

We now fix four base sites:
\[
a_X:=(1,0)\in\Lambda_A,\qquad a_Z:=(0,1)\in\Lambda_A,\qquad
q_0:=(0,0)\in\Lambda_Q,\qquad q_1:=(1,1)\in\Lambda_Q.
\]
Then every $X$-ancilla under row alternation ($y$ even) can be written uniquely as $a_X+(2i,2j)$ with $(i,j)\in G_0$, and every $Z$-ancilla ($y$ odd) as $a_Z+(2i,2j)$.
Likewise, every data qubit lies either in the $q_0$-coset or the $q_1$-coset:
\[
q_0+(2i,2j)\quad\text{or}\quad q_1+(2i,2j),\qquad (i,j)\in G_0.
\]

\begin{definition}[QC ring]
Let
\[
R:=\F_2[u^{\pm1},v^{\pm1}]/(u^{L_x/2}-1,\ v^{L_y/2}-1).
\]
We identify functions on $G_0$ with elements of $R$ in the usual way:
the monomial $u^i v^j$ denotes the indicator at $(i,j)\in G_0$ (extended periodically).
\end{definition}

\begin{figure}[t]
\centering
\begin{tikzpicture}[scale=0.95, every node/.style={font=\scriptsize}]
  % blocks
  \node[draw,rounded corners,minimum width=3.1cm,minimum height=1.2cm,align=center] (rowX) at (0,0) {$X$-ancillas\\(row alternation)\\one $G_0$-coset};
  \node[draw,rounded corners,minimum width=3.1cm,minimum height=1.2cm,align=center] (col0) at (5.0,0.8) {data coset\\$q_0+(2G_0)$};
  \node[draw,rounded corners,minimum width=3.1cm,minimum height=1.2cm,align=center] (col1) at (5.0,-0.8) {data coset\\$q_1+(2G_0)$};

  \draw[-{Latex[length=2.8mm]},thick] (rowX.east) -- node[above] {$h_0(u,v)$} (col0.west);
  \draw[-{Latex[length=2.8mm]},thick] (rowX.east) -- node[below] {$h_1(u,v)$} (col1.west);

  \node[align=left] at (0,-1.6) {\shortstack[c]{$H_X$ becomes a $1\times 2$ polynomial vector\\ over $R$.}};
  \node[align=left] at (5.0,-2.0) {\shortstack[l]{ $\text{full }H_X\text{ is obtained by expanding}$\\ $\text{these blocks over all }(i,j)\in G_0.$ }};
\end{tikzpicture}
\caption{Row alternation makes $\Lambda_X$ a single $G_0$-coset and $\Lambda_Q$ the disjoint union of two $G_0$-cosets.
Consequently, $H_X$ has a block-circulant form generated by a \emph{small} $1\times 2$ polynomial vector $(h_0,h_1)\in R^{1\times 2}$.}
\label{fig:qc-blockform}
\end{figure}
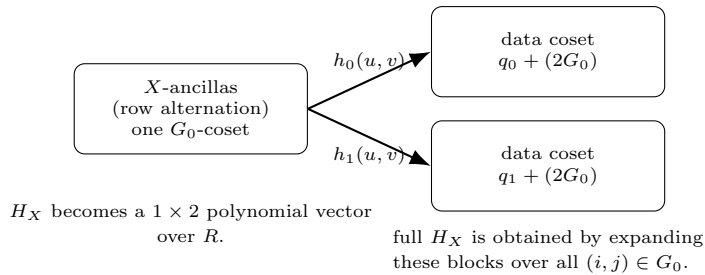

\subsection{From a word to a polynomial check vector}\label{subsec:word-to-qc-vector}

Fix a word $W$ and its support pattern $P(W)\subset\Z^2$ (Definition in Sec.~\ref{sec:construction}).
For each offset $Q=(Q_x,Q_y)\in P(W)$, consider the translate $a_X+Q$.
Since $a_X\in\Lambda_A$ and $Q$ has odd parity, $a_X+Q\in\Lambda_Q$.
Moreover, $a_X+Q$ lands in exactly one of the two data cosets $q_0+(2G_0)$ or $q_1+(2G_0)$, depending on $(Q_x,Q_y)\pmod 2$.

\begin{definition}[Row-alternating QC check vector]\label{def:qc-check-vector}
For each $Q\in P(W)$, define $(\sigma(Q),\delta(Q))$ as follows:
\begin{itemize}
\item $\sigma(Q)\in\{0,1\}$ is the unique index such that $a_X+Q\in q_{\sigma(Q)}+(2G_0)$;
\item $\delta(Q)=(\delta_x(Q),\delta_y(Q))\in G_0$ is defined by
\[
a_X+Q = q_{\sigma(Q)} + (2\delta_x(Q),\,2\delta_y(Q))\quad\text{in }G.
\]
\end{itemize}
Define two polynomials in $R$ by
\[
h_\sigma(u,v):=\sum_{Q\in P(W):\ \sigma(Q)=\sigma} u^{\delta_x(Q)}v^{\delta_y(Q)}\in R,\qquad \sigma\in\{0,1\}.
\]
We call $(h_0,h_1)\in R^{1\times 2}$ the \emph{QC check vector} of $(W,\text{row alternation})$ relative to the chosen bases.
\end{definition}

\begin{proposition}[Block-circulant expansion]\label{prop:block-circulant-expansion}
Under row alternation, the $X$-check matrix $H_X$ is the block-circulant matrix obtained by expanding the $1\times 2$ polynomial vector $(h_0,h_1)$ over $G_0$ as in Fig.~\ref{fig:qc-blockform}.
Equivalently, identifying row coefficients on $\Lambda_X$ with $f\in R$, the resulting syndrome-on-data vector is
\[
f\cdot (h_0,h_1) = (f h_0,\ f h_1)\in R\oplus R,
\]
where multiplication is in $R$.
An analogous statement holds for $H_Z$ with $(h_0,h_1)$ replaced by the corresponding $(g_0,g_1)$ computed from the base ancilla $a_Z$.
\end{proposition}

\subsection{Dependencies as an annihilator module and a computable $k$-formula}\label{subsec:annihilator-k}

Because we use \emph{every} ancilla as a check (either $X$ or $Z$), we have $m_X+m_Z=|\Lambda_A|=n$ rows total.
Thus, in this model the logical dimension $k$ is exactly the total number of independent stabilizer dependencies.

\begin{proposition}[Logical dimension equals number of stabilizer dependencies]\label{prop:k-equals-deps}
For any commuting CSS instance built on the checkerboard torus with one check per ancilla site,
\[
k = \dim\Ker(H_X^\top) + \dim\Ker(H_Z^\top).
\]
\end{proposition}

In the QC normal form, these dependency dimensions become small-ring kernel dimensions.

\begin{theorem}[QC dependency module for row alternation]\label{thm:qc-dependency-module}
Let $(h_0,h_1)\in R^{1\times 2}$ be the QC check vector for $H_X$ (Definition~\ref{def:qc-check-vector}).
Then
\[
\Ker(H_X^\top)\ \cong\ \{f\in R:\ f h_0=0\ \text{and}\ f h_1=0\}.
\]
In particular,
\[
\dim\Ker(H_X^\top)=\dim_{\F_2}\Ann_R(h_0,h_1),
\]
where $\Ann_R(h_0,h_1)=\{f\in R:\ f h_0=f h_1=0\}$ is the annihilator module.
An analogous statement holds for $H_Z$ with its QC vector $(g_0,g_1)$.
\end{theorem}

\subsection{Closed-form prediction of the $k$-collapse for $W=\word{NE$^2$NE$^2$N}$}\label{subsec:kcollapse-proof}

We now apply the framework to the case-study word
\[
W=\word{NE$^2$NE$^2$N},\qquad
P(W)=\{(0,1),(1,2),(3,2),(4,3),(5,4),(7,4),(8,5)\}
\]
as drawn in Fig.~\ref{fig:fivewords}.
For row alternation, the QC check vector $(h_0,h_1)$ can be computed explicitly.

\begin{example}[QC check vector for $\word{NE$^2$NE$^2$N}$ under row alternation]\label{ex:qc-vector-case}
Using the bases $(a_X,q_0,q_1)$ from Sec.~\ref{subsec:coarse-qc}, one finds
\[
h_0(u,v)=u v + u^2 v + u^3 v^2 + u^4 v^2,\qquad
h_1(u,v)=1 + u^2 v + u^4 v^2,
\]
in the ring $R=\F_2[u^{\pm1},v^{\pm1}]/(u^{L_x/2}-1,\ v^{L_y/2}-1)$.
\end{example}

\begin{proof}
We list the offsets $Q\in P(W)$ and determine whether $a_X+Q$ lands in the $q_0$ or $q_1$ data coset.
A direct parity check shows
\[
a_X+(1,2),\ a_X+(3,2),\ a_X+(5,4),\ a_X+(7,4)\in q_0+(2G_0),
\]
\[
a_X+(0,1),\ a_X+(4,3),\ a_X+(8,5)\in q_1+(2G_0).
\]
Writing each target uniquely as $q_\sigma+(2\delta_x,2\delta_y)$ produces the exponents
\[
(1,2)\mapsto u^1v^1,\ (3,2)\mapsto u^2v^1,\ (5,4)\mapsto u^3v^2,\ (7,4)\mapsto u^4v^2
\]
for $h_0$, and
\[
(0,1)\mapsto 1,\ (4,3)\mapsto u^2v^1,\ (8,5)\mapsto u^4v^2
\]
for $h_1$, yielding the stated polynomials.
\end{proof}

The key simplification is that $h_0$ has \emph{paired} $v$-exponents, so it is annihilated by the ``sum over all $u$'' idempotent.

\begin{definition}[Horizontal sum idempotent]\label{def:Su}
Let
\[
S_u := 1+u+u^2+\cdots+u^{L_x/2-1}\ \in\ R.
\]
Then $u^t S_u=S_u$ for all $t$, so $S_u$ is invariant under horizontal shifts.
\end{definition}

\begin{lemma}[A universal annihilation of $h_0$]\label{lem:Su-kills-h0}
For the case-study word, $S_u\,h_0=0$ in $R$ for every $(L_x,L_y)$.
Moreover,
\[
S_u\,h_1 = S_u\,(1+v+v^2).
\]
\end{lemma}

Therefore, $X$-check dependencies are exactly controlled by the 1D cyclotomic factor $1+v+v^2$ and its interaction with the vertical period $v^{L_y/2}-1$.

\begin{theorem}[Exact $k$-collapse criterion for row alternation]\label{thm:kcollapse-criterion}
Let $W=\word{NE$^2$NE$^2$N}$ and use the row-alternating layout on $G=\Z_{L_x}\times\Z_{L_y}$.
Assume $L_x,L_y$ are even so the checkerboard model is well-defined.
Then
\[
\dim\Ker(H_X^\top)=\deg\gcd\big(1+v+v^2,\ v^{L_y/2}-1\big),
\qquad
\dim\Ker(H_Z^\top)=\deg\gcd\big(1+v+v^2,\ v^{L_y/2}-1\big),
\]
and hence
\[
k \;=\; 2\,\deg\gcd\big(1+v+v^2,\ v^{L_y/2}-1\big).
\]
In particular, $k=4$ if $3\mid (L_y/2)$ (equivalently $6\mid L_y$), and $k=0$ otherwise.
\end{theorem}

\begin{remark}[Interpretation of Table~\ref{tab:kcollapse}]
Theorem~\ref{thm:kcollapse-criterion} explains the observed alternation of $k$ in Table~\ref{tab:kcollapse} without constructing $H_X,H_Z$ explicitly.
For thin rectangles $(L_x,L_y)=(2d,d)$ one has $L_y/2=d/2$, so $k=4$ exactly when $3\mid (d/2)$, i.e.\ $6\mid d$.
This is precisely the commensurability mechanism that also stabilizes $k$ on the $(12m,6m)$ family.
\end{remark}

% ==========================================================
\section{Enumeration and canonicalization of coset-constant layouts}\label{sec:layout-enum}
% ==========================================================
% Suggested placement:
% Put this section immediately after Sec.~\ref{sec:invariants} (after Theorem~\ref{thm:layout} and Table~\ref{tab:lattices2}).
% It turns ``layout freedom'' into a countable search space and explains how to avoid redundant layout scans.

Theorem~\ref{thm:layout} states that commutation-compatible layouts are constant on cosets of the odd-multiplicity lattice $L(W)$.
This section makes that freedom explicit: it counts coset-constant layouts, describes natural equivalences (global $X\!\leftrightarrow\! Z$ and translations), and gives a canonical representative rule that supports reproducible layout scans.

\subsection{Layouts as $\F_2$-labelings of ancilla cosets}

Fix a word $W$ and assume $(L_x,L_y)$ are such that the image of $L(W)$ in the torus is a subgroup that preserves the ancilla sublattice.
Let $\mathcal{C}_A(W)$ denote the finite set of ancilla cosets:
\[
\mathcal{C}_A(W):=\Lambda_A / L(W),
\qquad c:=|\mathcal{C}_A(W)|.
\]
A coset-constant layout is equivalently a function
\[
\alpha:\mathcal{C}_A(W)\to\{X,Z\}\ \cong\ \F_2.
\]

\begin{proposition}[Raw layout count]\label{prop:raw-layout-count}
The number of commutation-compatible coset-constant layouts on $\Lambda_A$ equals $2^c$.
Modulo the global swap $X\leftrightarrow Z$, the number of distinct layouts equals $2^{c-1}$.
\end{proposition}

\begin{example}[Two- and four-coset words]\label{ex:layout-counts}
If $L(W)=2\Z\times 2\Z$, then $c=2$ (Fig.~\ref{fig:cosets}, left), so there is only one nontrivial layout up to the global swap: one coset is $X$ and the other is $Z$.
If $L(W)=\Span\{(4,0),(2,2)\}$, then $c=4$ (Fig.~\ref{fig:cosets}, right), giving $2^{4-1}=8$ layouts up to global swap.
\end{example}

\subsection{Translation equivalence and canonical representatives}

Even when two layouts are different as functions on $\mathcal{C}_A(W)$, they can yield isomorphic codes if one is obtained from the other by a translation symmetry of the torus that preserves $W$ (equivalently, preserves the pattern up to a translate).
Such equivalences matter for searches: scanning all $2^{c-1}$ layouts may be redundant.

Let $\mathcal{T}$ be the subgroup of even translations of $G$ (those preserving $\Lambda_A$) that act on $\mathcal{C}_A(W)$ by permuting cosets.
(Concretely, $\tau\in\mathcal{T}$ sends the coset of $a$ to the coset of $a+\tau$.)

\begin{proposition}[Translation-equivalent layouts yield isomorphic codes]\label{prop:layout-translation-equivalence}
If two coset-constant layouts $\alpha,\alpha'$ satisfy $\alpha'=\alpha\circ \pi$ for some coset permutation $\pi$ induced by a translation of the torus, then the resulting CSS codes are isomorphic by a qubit relabeling (translation of coordinates).
\end{proposition}

\begin{definition}[Canonical layout representative]\label{def:canonical-layout}
Fix an ordering of the cosets in $\mathcal{C}_A(W)$ (e.g.\ lexicographic by a chosen coset representative).
Given a layout $\alpha$, define $\mathrm{can}(\alpha)$ as follows:
apply every translation-induced coset permutation in $\mathcal{T}$ to $\alpha$, then (if desired) apply the global flip so that the first coset is labeled $X$, and finally choose the lexicographically smallest bitstring among the resulting layouts.
\end{definition}

\begin{remark}[Practical use in searches]
Definition~\ref{def:canonical-layout} is a concrete, implementation-friendly way to avoid redundant layout scans.
In a script, one simply generates all $2^{c-1}$ labelings, canonicalizes each, and keeps one representative per canonical string.
\end{remark}

% \begin{figure}[t]
% \centering
% \begin{tikzpicture}[
%   every node/.style={font=\scriptsize},
%   box/.style={draw,rounded corners,minimum width=1.0cm,minimum height=0.8cm,
%               align=center,inner sep=2pt},
%   >=Latex
% ]
% % --- 2x2 coset grid as a matrix (stable spacing) ---
% \matrix (M) [matrix of nodes,
%   nodes={box},
%   column sep=7mm,
%   row sep=6mm
% ] {
%   |[fill=gray!15]| A & |[fill=gray!35]| B \\
%   |[fill=gray!15]| C & |[fill=gray!35]| D \\
% };

% % --- arrows (use east/west anchors so lines don't cut through boxes) ---
% \draw[->,thick] (M-1-1.east) -- node[above, yshift=1pt] {$\tau$} (M-1-2.west);
% \draw[->,thick] (M-2-1.east) -- node[above, yshift=1pt] {$\tau$} (M-2-2.west);

% % --- right-hand explanation as ONE wrapped node (prevents overlap) ---
% \node[anchor=west,align=left,text width=6.2cm] at ($(M.east)+(1.2,0)$) {%
% Example: $L(W)=\langle(4,0),(2,2)\rangle$ gives four ancilla cosets A--D.\\
% An even translation can permute cosets (e.g.\ swap A$\leftrightarrow$B and C$\leftrightarrow$D).%
% };

% \end{tikzpicture}
% \caption{Abstract coset-permutation viewpoint for four-coset words.
% Coset-constant layouts are labelings of A--D by $X/Z$.
% Translation symmetries act by permuting cosets, so many labelings can be equivalent up to coordinate relabeling (Proposition~\ref{prop:layout-translation-equivalence}).}
% \label{fig:layout-coset-permutation}
% \end{figure}

% ==========================================================
\section{Row-Periodic Quasi-Cyclic Reduction via Annihilators}\label{sec:qc}
% ==========================================================

We now add a refinement of the quasi-cyclic viewpoint tailored to common layouts (especially row alternation): for a $2\times2$ unit cell, rank deficiency and stabilizer dependencies reduce to annihilators in a small polynomial ring.

\subsection{The $2\times2$ unit cell and group algebra}
Assume $L_x,L_y$ even and define the reduced torus
\[
G_0 := \Z_{L_x/2}\times \Z_{L_y/2}.
\]
Let
\[
R := \F_2[X^{\pm 1},Y^{\pm 1}]/(X^{L_x/2}-1,\;Y^{L_y/2}-1)
\]
be the group algebra of $G_0$.

Row-periodic layouts (including row alternation) yield block-circulant check matrices generated by constant-size polynomial rows.
For row alternation, there is one $X$-check type and one $Z$-check type per cell, represented by
\[
\bm{h}_X=(h_{X,0},h_{X,1})\in R^{1\times 2},\qquad
\bm{h}_Z=(h_{Z,0},h_{Z,1})\in R^{1\times 2},
\]
where the two components correspond to the two data-qubit types per cell.

\subsection{Annihilators and stabilizer dependencies}
\begin{definition}[Annihilator]\label{def:annih}
For $\bm{h}=(h_0,h_1)\in R^{1\times 2}$, define
\[
\mathrm{Ann}(\bm{h}) := \{f\in R:\; f h_0=0 \ \text{and}\ f h_1=0\}.
\]
\end{definition}

\begin{proposition}[Rank deficiency from annihilators]\label{prop:ann_rank}
Let $H$ be the block-circulant matrix over $\F_2$ generated by translating $\bm{h}\in R^{1\times 2}$ over $G_0$.
Then $\dim_{\F_2}\Ker(H^\top)=\dim_{\F_2}\mathrm{Ann}(\bm{h})$ and
\[
\rank(H)=|G_0|-\dim_{\F_2}\mathrm{Ann}(\bm{h}).
\]
\end{proposition}

\begin{corollary}[A $k$-formula for row alternation]\label{cor:k_ann}
In the row-alternating layout (one check type per cell for each Pauli type),
\[
k=\dim_{\F_2}\mathrm{Ann}(\bm{h}_X)+\dim_{\F_2}\mathrm{Ann}(\bm{h}_Z).
\]
\end{corollary}

% ==========================================================
\section{Worked Families and a Detailed Case Study}\label{sec:families}
% ==========================================================

This section illustrates how the framework is used in practice:
\[
W\ \longrightarrow\ P(W)\ \longrightarrow\ L(W)\ \longrightarrow\ \text{layout class}\ \longrightarrow\ \text{finite-torus behavior}.
\]

\subsection{Benchmark words through the invariant lens}
The words $\word{NE$^2$N}$, $\word{NE3N}$, $\word{N$^2$E$^2$N$^2$}$, and $\word{N2E3N2}$ are central in \cite{Geher2025Directional}.
Table~\ref{tab:lattices2} highlights a structural distinction:
$\word{NE$^2$N}$ and $\word{N$^2$E$^2$N$^2$}$ have $L(W)=2\Z\times2\Z$ (two ancilla cosets), while $\word{NE$^3$N}$ and $\word{N$^2$E$^3$N$^2$}$ have $L(W)=\Span\{(4,0),(2,2)\}$ (four ancilla cosets).
The latter class admits more layout freedom (Fig.~\ref{fig:cosets}), which is often useful when searching for good finite-size instances.

\subsection{Case study: the word $W=\word{NE$^2$NE$^2$N}$ and its four cosets}
We now focus on
\[
W=\word{NE$^2$NE$^2$N}=\word{N\,E\,E\,N\,E\,E\,N},
\]
with support pattern shown in Fig.~\ref{fig:fivewords}.
From Table~\ref{tab:lattices2},
\begin{equation}
L(W)=\Span_{\Z}\{(4,0),(2,2)\},
\end{equation}
so layouts must be constant on four ancilla cosets.

A useful explicit coset invariant is the residue pair
\begin{equation}
\big(y \bmod 2,\;\; (x-y)\bmod 4\big),
\end{equation}
which is constant on each $L(W)$-coset. On ancilla sites ($x+y$ odd) this yields exactly four classes, labeled A--D in Fig.~\ref{fig:cosetlabels}.
Row alternation uses only $y\bmod 2$, while general coset-constant layouts may depend on both residues.

\begin{figure}[t]
\centering
\begin{tikzpicture}[scale=0.55]
  \draw[step=1,gray!20,very thin] (0,0) grid (8,6);

  % label 4 cosets for L(W)=<(4,0),(2,2)>
  \foreach \x in {0,...,7}{
    \foreach \y in {0,...,5}{
      \pgfmathtruncatemacro{\p}{mod(\x+\y,2)}
      \ifnum\p=1
        \pgfmathtruncatemacro{\ym}{mod(\y,2)}
        \pgfmathtruncatemacro{\xm}{mod(\x-\y+100,4)}
        \def\lab{}
        \ifnum\ym=0
          \ifnum\xm=1 \def\lab{A}\fi
          \ifnum\xm=3 \def\lab{B}\fi
        \else
          \ifnum\xm=1 \def\lab{C}\fi
          \ifnum\xm=3 \def\lab{D}\fi
        \fi
        \fill[gray!25] (\x+0.5,\y+0.5) circle (2.0pt);
        \node at (\x+0.5,\y+0.5) {\scriptsize \lab};
      \fi
    }
  }
\end{tikzpicture}
\caption{Ancilla cosets for $L(W)=\Span\{(4,0),(2,2)\}$ for the case-study word $W=\word{NE$^2$NE$^2$N}$.
Two ancillas $(x,y)$ and $(x',y')$ lie in the same $L(W)$-coset iff the pair of residues $\big(y\bmod 2,\ (x-y)\bmod 4\big)$ matches.
Any commutation-compatible layout must be constant on each labeled class (Theorem~\ref{thm:layout}), but the assignment of $X$ versus $Z$ across A--D is otherwise free up to global symmetries.}
\label{fig:cosetlabels}
\end{figure}

\subsection{A structured torus family and certified nonzero $k$}
We study the family $G_m=\Z_{12m}\times \Z_{6m}$ with the row-alternating layout (Fig.~\ref{fig:checkerboard}).
Then $n=\tfrac12(12m)(6m)=36m^2$.
In this family we can certify nonzero $k$ by explicit stabilizer dependencies and exhibit explicit commuting operators that upper-bound distance.

\begin{proposition}[Two stabilizer dependencies per type on $12m\times 6m$]\label{prop:deps}
On $G_m$ with the row-alternating layout, there exist two independent $\F_2$ relations among the $X$-checks and two among the $Z$-checks.
Consequently,
\[
\rank(H_X)\le \#\Lambda_X-2,\qquad \rank(H_Z)\le \#\Lambda_Z-2,
\]
and therefore $k\ge 4$ for all $m$.
\end{proposition}

The dependencies are not mysterious algebraic accidents: they come from parity cancellations that are visible at the level of the word. 
Very roughly, when one multiplies (adds over $\F_2$) all $X$-checks whose anchors lie in certain $y$-residue classes modulo $6$, each data qubit is hit an even number of times because the $y$-coordinates in $P(W)$ come in repeated pairs mod $6$.
This is why the same word on an incommensurate height can lose those dependencies and collapse to $k=0$ under the same row-alternating layout (as seen in the thin-rectangle examples of Sec.~\ref{sec:numerics}): the torus no longer supports the residue-class partition underlying the cancellation. Fig.~\ref{fig:depscartoon} shows this case. 

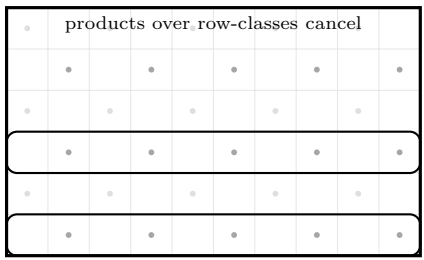
\begin{figure}[t]
\centering
\begin{tikzpicture}[scale=0.55]
  \draw[step=1,gray!20,very thin] (0,0) grid (10,6);
  \draw[very thick] (0,0) rectangle (10,6);
  % ancillas (odd parity) and mark X-rows y even
  \foreach \x in {0,...,9}{
    \foreach \y in {0,...,5}{
      \pgfmathtruncatemacro{\p}{mod(\x+\y,2)}
      \ifnum\p=1
        \pgfmathtruncatemacro{\py}{mod(\y,2)}
        \ifnum\py=0
          \fill[gray!70] (\x+0.5,\y+0.5) circle (2.0pt);
        \else
          \fill[gray!25] (\x+0.5,\y+0.5) circle (2.0pt);
        \fi
      \fi
    }
  }
  \draw[thick,rounded corners] (0,0) rectangle (10,1);
  \draw[thick,rounded corners] (0,2) rectangle (10,3);
  \node at (5,5.6) {\scriptsize products over row-classes cancel};
\end{tikzpicture}
\caption{Illustration of Proposition~\ref{prop:deps}: on commensurate tori, products of checks over entire residue classes of anchor rows can cancel, producing stabilizer dependencies (formalized in Appendix~\ref{app:deps}).}
\label{fig:depscartoon}
\end{figure}

\subsection{An explicit commuting-operator motif}
Let $t=(12,6)$ and $r=(4,2)$ modulo $(12m,6m)$.
For a data site $p\in\Lambda_Q$, define
\begin{equation}
S(p)=\{p+jt,\ p+r+jt:\ j=0,1,\dots,m-1\}.
\end{equation}
Geometrically, $S(p)$ consists of two interleaved length-$m$ orbits under translation by $t$, separated by the fixed displacement $r$ (Fig.~\ref{fig:motif}).
The usefulness of this construction is that checking commutation reduces to an overlap-parity count between $S(p)$ and translated copies of $P(W)$; for suitable $p$, each check intersects $S(p)$ in an even number of sites.

\begin{proposition}[Weight-$2m$ commuting operators]\label{prop:2m}
For suitable choices of $p$, the Pauli operators $X$ supported on $S(p)$ and $Z$ supported on $S(p)$ commute with all opposite-type checks.
In particular, for this family $d\le 2m$.
\end{proposition}

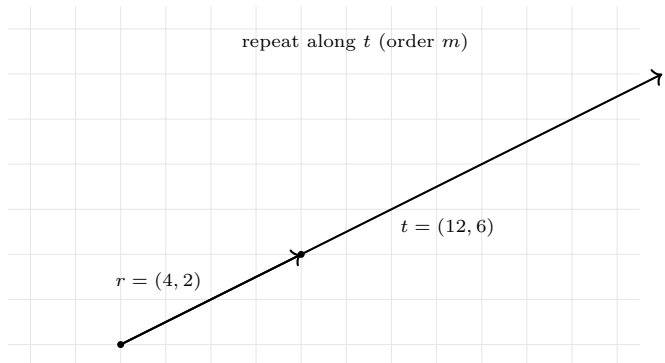
\begin{figure}[t]
\centering
\begin{tikzpicture}[scale=0.60]
  \draw[step=1,gray!20,very thin] (-0.5,-0.5) grid (13.5,7.5);
  \fill (2,0) circle (2.2pt);
  \fill (6,2) circle (2.2pt);
  \draw[thick,->] (2,0)--(6,2) node[midway,above left] {\scriptsize $r=(4,2)$};
  \draw[thick,->] (2,0)--(14,6) node[midway,below right] {\scriptsize $t=(12,6)$};
  \node at (7.2,6.7) {\scriptsize repeat along $t$ (order $m$)};
\end{tikzpicture}
\caption{The $(r,t)$ motif in Proposition~\ref{prop:2m}: two sites per step along a length-$m$ orbit, giving weight $2m$ and hence a linear-in-$m$ distance upper bound.}
\label{fig:motif}
\end{figure}

\subsection{Exact criterion for dimension collapse on thin rectangles}\label{sec:kc}
A central practical phenomenon is boundary-condition sensitivity of $k$ under a fixed word and layout.
For the case-study word and row alternation on thin rectangles $(L_x,L_y)=(2d,d)$, the observed behavior admits an exact divisibility criterion.

\begin{theorem}[Exact $k$-collapse criterion for $W=\word{NE$^2$NE$^2$N}$]\label{thm:kc-collapse}
Let $W=\word{NE$^2$NE$^2$N}$ and consider the directional CSS instance on $(L_x,L_y)=(2d,d)$ with $d$ even and the row-alternating layout.
Then
\[
k=
\begin{cases}
4,& 6\mid d,\\
0,& 6\nmid d.
\end{cases}
\]
\end{theorem}

The proof uses the annihilator reduction from Section~\ref{sec:qc} and is given in compressed form in Appendix~\ref{app:kc}.

% ==========================================================
\section{Numerical Experiments and Word Search}\label{sec:numerics}
% ==========================================================

This section explains numerical-results obtained in a reproducible way and how to interpret the kinds of finite-size effects that arise for route-generated, translation-invariant codes.
The aim is not to claim best-in-class performance for the illustrative examples, but to provide a repeatable baseline for (i) validating a proposed word/layout instance and (ii) searching for promising new words under a fixed footprint and connectivity model.
Throughout we emphasize that the quantities reported here are \emph{static} code parameters of the CSS check matrices; circuit-level performance and thresholds require an explicit extraction circuit and a decoder model, which are outside our present scope.

\subsection{Reproducible workflow for validation and word search.}%What we measure.}
For a fixed $(W,L_x,L_y,\alpha)$, we build sparse $H_X,H_Z$, verify commutation ($H_XH_Z^\top=0$), compute $(n,k)$ exactly by rank over $\F_2$, and estimate $d_X,d_Z$ by exhaustive search up to a cutoff weight $w_{\max}$.
If no nontrivial logical is found up to $w_{\max}$, we report $d>w_{\max}$.
In practice, the exhaustive search cost grows quickly with $n$ and $w_{\max}$, so entries of the form $d>w_{\max}$ should be read as \emph{screening lower bounds} rather than definitive distances.
A typical workflow is to use a small cutoff to shortlist candidates and then rerun only the shortlist with higher cutoffs or specialized distance routines. The end-to-end workflow used to generate the tables in this section is summarized in Fig.~\ref{fig:numerics-workflow}.

%\paragraph{Script-based reproducibility.}
All tables were generated using a companion script \cite{rowshan2026dircoode}.
The script parses compressed words (e.g.\ \word{NE$^2$N$^2$E$^2$N}), constructs $P(W)$ from Lemma~\ref{lem:offset}, builds $H_X,H_Z$ on the checkerboard torus, and runs small-weight searches.
The same script can enumerate words up to symmetry using Proposition~\ref{prop:equiv} and can filter by simple route constraints (no immediate backtracking, distinct offsets, and/or ``full support'' without cancellations).
We stress that the tables below use the \emph{row-alternating} layout for concreteness; for words whose odd-difference lattice has more than two ancilla cosets (e.g.\ $L(W)=\Span\{(4,0),(2,2)\}$ in Table~\ref{tab:lattices2}), it is often beneficial to also enumerate all coset-constant layouts allowed by Theorem~\ref{thm:layout} and repeat the same screening procedure.

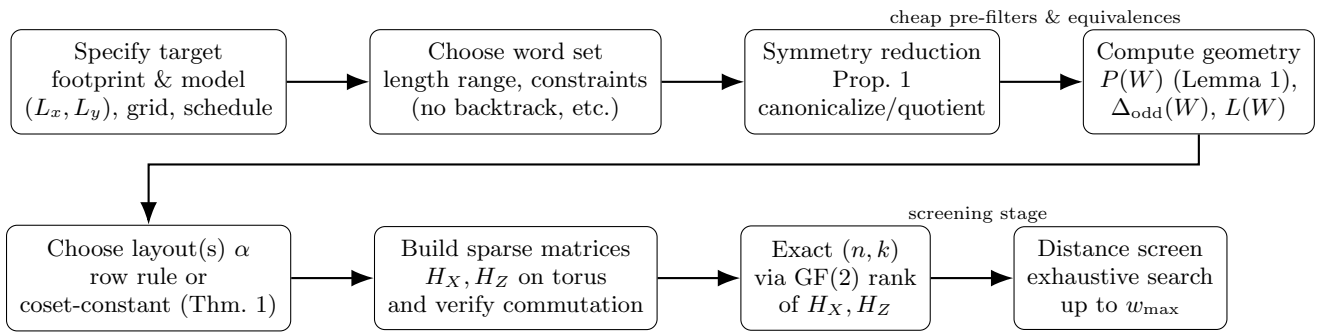
\begin{figure*}[t]
\centering
\begin{tikzpicture}[
  node distance=1.2cm and 1.1cm,
  box/.style={draw,rounded corners,align=center,inner sep=5pt,minimum height=10mm},
  opt/.style={draw,rounded corners,dashed,align=center,inner sep=5pt,minimum height=10mm},
  arrow/.style={-{Latex[length=3mm]},thick}
]
% --- top row ---
\node[box] (spec) {Specify target\\footprint \& model\\$(L_x,L_y)$, grid, schedule};
\node[box,right=of spec] (words) {Choose word set\\length range, constraints\\(no backtrack, etc.)};
\node[box,right=of words] (sym) {Symmetry reduction\\Prop.~\ref{prop:equiv}\\canonicalize/quotient};
\node[box,right=of sym] (pattern) {Compute geometry\\$P(W)$ (Lemma~\ref{lem:offset}),\\$\Delta_{\rm odd}(W)$, $L(W)$};

% --- bottom row ---
\node[box,below=of spec] (layout) {Choose layout(s) $\alpha$\\row rule or\\coset-constant (Thm.~\ref{thm:layout})};
\node[box,right=of layout] (build) {Build sparse matrices\\$H_X,H_Z$ on torus\\and verify commutation};
\node[box,right=of build] (rank) {Exact $(n,k)$\\via GF$(2)$ rank\\of $H_X,H_Z$};
\node[box,right=of rank] (dist) {Distance screen\\exhaustive search\\up to $w_{\max}$};

% % --- outputs / refine ---
% \node[box,below=of rank, xshift=2.7cm] (report) {Report \& store\\tables + selected patterns\\shortlist candidates};
% \node[opt,right=of report] (refine) {Refine shortlist\\increase $w_{\max}$,\\scan nearby $(L_x,L_y)$,\\try more layouts};

% --- arrows ---
\draw[arrow] (spec) -- (words);
\draw[arrow] (words) -- (sym);
\draw[arrow] (sym) -- (pattern);
%\draw[arrow] %(pattern) |- (layout);
\path (pattern.south) ++(0,-0.4) coordinate (midbend);
\draw[arrow] (pattern.south) -- (midbend) -| (layout.north);

\draw[arrow] (layout) -- (build);
\draw[arrow] (build) -- (rank);
\draw[arrow] (rank) -- (dist);
% \draw[arrow] (dist) |- (report);
% \draw[arrow] (report) -- (refine);

% small labels
\node[align=center] at ($(sym)!0.5!(pattern) + (0,0.85)$) {\scriptsize cheap pre-filters \& equivalences};
\node[align=center] at ($(rank)!0.5!(dist) + (0,0.85)$) {\scriptsize screening stage};

\end{tikzpicture}
\caption{Reproducible pipeline used in Sec.~\ref{sec:numerics}. The workflow is identical for validating a fixed $(W,L_x,L_y,\alpha)$ instance or for searching over many words. A practical search runs cheap filters and symmetry quotienting early, computes $(n,k)$ exactly, and uses a small-weight distance screen to build a shortlist before investing in deeper distance calculations or broader layout/boundary scans.}
\label{fig:numerics-workflow}
\end{figure*}

\subsection{Case study family and its scaling signals.}
For $W=\word{NE$^2$NE$^2$N}$ on $12m\times 6m$ tori with row alternation, Table~\ref{tab:case-family} summarizes $(n,k)$ and distance searches.
Two trends are worth highlighting.
First, the observed constant $k=4$ across $m=1,2,3$ is consistent with the analytical stabilizer-dependency mechanism established earlier: Proposition~\ref{prop:deps} guarantees two independent dependencies among the $X$-checks and two among the $Z$-checks on this commensurate torus family, yielding $k\ge 4$ for all $m$.
In other words, the \emph{dimension is protected} by a structural dependency that is robust under scaling along $(12m,6m)$, not a coincidence of small size.
Second, the distances for $m=1,2$ saturate the explicit commuting-operator construction: Proposition~\ref{prop:2m} gives a family of weight-$2m$ commuting operators, hence $d\le 2m$, while Table~\ref{tab:case-family} shows $d_X=d_Z=2m$ for $m=1,2$.
For $m=3$ the table only reports $d_X,d_Z>3$ because the search cutoff was $w_{\max}=3$; combining this with Proposition~\ref{prop:2m} yields the informative bracket $4\le d \le 6$ for $m=3$.
A natural next numerical step (and one we recommend when reporting this family on larger footprints) is to increase $w_{\max}$ for this single word on the structured tori to test the working hypothesis suggested by the first two points, namely that $d=2m$ may hold along $(12m,6m)$ for this layout.

\begin{table}[t]
\centering
\caption{Case-study family $W=\word{NE$^2$NE$^2$N}$ on $12m\times 6m$ tori (row-alternating layout).
Entries $d>3$ mean no nontrivial logical of weight $\le 3$ was found, so $d\ge 4$.}
\label{tab:case-family}
\begin{tabular}{@{}cccccc@{}}
\toprule
$m$ & $L_x\times L_y$ & $n$ & $k$ & $d_X$ & $d_Z$ \\
\midrule
1 & 12$\times$6 & 36  & 4 & 2   & 2 \\
2 & 24$\times$12 & 144 & 4 & 4   & 4 \\
3 & 36$\times$18 & 324 & 4 & $>3$ & $>3$ \\
\bottomrule
\end{tabular}
\end{table}

\subsection{Boundary-condition sensitivity as a design constraint.}
Keeping the same word and layout but changing boundary conditions can collapse $k$.
Table~\ref{tab:kcollapse} illustrates this for thin rectangles $(L_x,L_y)=(2d,d)$, where $n=d^2$ makes the footprint attractive.
The striking feature is that $k$ alternates between $4$ and $0$ as $d$ changes: $k=4$ for $d=6,12,18$ but $k=0$ for $d=8,10,14$ under the same row-alternating layout.
This is a concrete instance of a familiar quasi-cyclic phenomenon: for translation-invariant codes, the ranks of $H_X$ and $H_Z$ (hence $k$) can depend strongly on the torus periods.
Here the effect has a particularly simple interpretation.
The stabilizer dependencies that certify $k\ge 4$ on $(12m,6m)$ are built from cancellations over residue classes modulo $6$ in the vertical coordinate (Appendix~\ref{app:deps}); these cancellations close consistently only when $L_y$ is a multiple of $6$.
When $L_y$ is not divisible by $6$, those specific global dependencies disappear, the check matrices become full-rank in the relevant sense, and the resulting code encodes no logical qubits ($k=0$).
Rather than viewing such instances as failures, Table~\ref{tab:kcollapse} should be read as a practical message for hardware-driven design: \emph{boundary conditions are part of the code specification} for route-generated quasi-cyclic constructions, and one should scan over $(L_x,L_y)$ (or over allowed shear/quotient choices) rather than fixing a single aspect ratio and assuming the dimension will be stable.
It also motivates exploring alternative allowed layouts beyond row alternation: because $W=\word{NE$^2$NE$^2$N}$ has $L(W)=\Span\{(4,0),(2,2)\}$, Theorem~\ref{thm:layout} permits a four-coset family of layouts, some of which can restore $k>0$ on sizes where the row-alternating layout collapses.

\begin{table}[t]
\centering
\caption{Dimension collapse for $W=\word{NE$^2$NE$^2$N}$ on thin rectangles $(L_x,L_y)=(2d,d)$ (row-alternating layout).}
\label{tab:kcollapse}
\begin{tabular}{@{}ccccc@{}}
\toprule
$d$ & $L_x\times L_y$ & $n$ & $k$ \\
\midrule
6  & 12$\times$6  & 36  & 4 \\
8  & 16$\times$8  & 64  & 0 \\
10 & 20$\times$10 & 100 & 0 \\
12 & 24$\times$12 & 144 & 4 \\
14 & 28$\times$14 & 196 & 0 \\
18 & 36$\times$18 & 324 & 4 \\
\bottomrule
\end{tabular}
\end{table}

\subsection{Understanding the Search Table: Insights and Limitations.}
To show how a ``new code search'' can be documented, we scanned words of length $4\le w\le 8$ on the $16\times 8$ torus,
quotiented by symmetry (fixing the first letter $N$), disallowing immediate backtracking, and keeping only commuting CSS instances under row alternation.
Table~\ref{tab:search16x8} lists the top candidates found under this setting.
The goal here is not to optimize a single metric, but to illustrate what a transparent reporting pipeline looks like and what kinds of tradeoffs appear even at modest sizes.

A first observation is that a fixed footprint can admit many distinct commuting words with markedly different \emph{dimension} $k$ at the same $n$.
For example, \word{NE$^2$N$^2$E$^2$N} attains $k=18$ at $n=64$, while several other length-$8$ words yield $k=10$ or $k=6$.
High $k$ at fixed $n$ typically reflects additional stabilizer dependencies (fewer independent checks), and in quasi-cyclic settings such dependencies are often tied to symmetry and commensurability effects similar to those seen in Table~\ref{tab:kcollapse}.
A second observation is that the reported distances for most candidates are modest ($d_X=d_Z=4$ under the cutoff), while \word{NES$^2$EN} stands out with $d_X,d_Z>4$ at the chosen cutoff, making it a natural candidate for deeper distance tests.
At the same time, one should be cautious in over-interpreting the distance column at this stage: for entries with $d>4$ the true distance could be $5$ or substantially larger, and conversely distance can change under different layouts or boundary choices.
This is why we recommend a two-stage procedure: use tables like Table~\ref{tab:search16x8} to shortlist, then rerun the shortlist with larger $w_{\max}$ and with additional allowed layouts (coset-constant layouts from Theorem~\ref{thm:layout}) to check robustness.

Finally, the support-pattern sketches in Fig.~\ref{fig:search16x8-patterns} provide geometric intuition about why very different words can survive the same footprint and layout filter.
The examples include both ``snake-like'' patterns and patterns with pronounced turns or near-loops; empirically, these geometric features correlate with how often translated checks overlap and therefore with the likelihood of additional dependencies and/or short logical operators.
From an implementation perspective, there is also a direct engineering tradeoff hidden behind the word length $w$: longer words correspond to deeper per-check measurement schedules (more two-qubit interactions), which may hurt circuit-level performance even if the static parameters look attractive.
For this reason, when using Table~\ref{tab:search16x8} as a discovery tool, we recommend reporting $w$ together with the effective support size $|P(W)|$ (which may be smaller than $w$ if offsets cancel modulo $2$), and treating short, robust-to-boundaries candidates as especially valuable starting points for circuit-level studies.

\begin{table}[t]
\centering
\caption{Top candidates from a scan on the $16\times 8$ torus (row-alternating layout).
Entries $d>4$ mean no nontrivial logical of weight $\le 4$ was found, so $d\ge 5$.}
\label{tab:search16x8}
\begin{tabular}{@{}cccccc@{}}
\toprule
word $W$ & $w$ & $n$ & $k$ & $d_X$ & $d_Z$ \\
\midrule
NES$^2$EN      & 6 & 64 & 6  & $> 4$ & $> 4$ \\
NE$^2$N$^2$E$^2$N    & 8 & 64 & 18 & 4 & 4 \\
N$^2$ENW$^2$NE    & 8 & 64 & 10 & 4 & 4 \\
N$^2$ESW$^2$SE    & 8 & 64 & 10 & 4 & 4 \\
N$^3$E$^2$NW$^2$     & 8 & 64 & 10 & 4 & 4 \\
NENE$^2$NEN    & 8 & 64 & 10 & 4 & 4 \\
NENENWNW    & 8 & 64 & 10 & 4 & 4 \\
NENESWSW    & 8 & 64 & 10 & 4 & 4 \\
N$^2$E$^2$N$^2$      & 6 & 64 & 6  & 4 & 4 \\
NENW$^2$NES    & 8 & 64 & 6  & 4 & 4 \\
\bottomrule
\end{tabular}
\end{table}

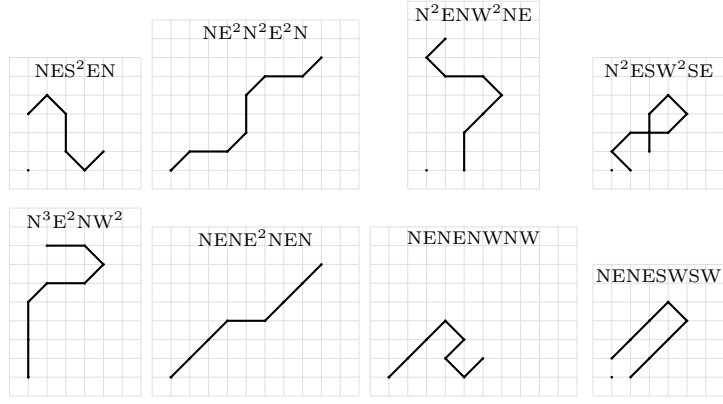
\begin{figure*}[t]
\centering
\begin{tabular}{cccc}
\begin{tikzpicture}[scale=0.25]
  \draw[step=1,gray!25,very thin] (-1,-1) grid (6,6);
  \fill (0,0) circle (2.0pt);
  \foreach \x/\y in {0/3,1/4,2/3,2/1,3/0,4/1} { \fill (\x,\y) circle (2.0pt); }
  \draw[thick] (0,3)--(1,4)--(2,3)--(2,1)--(3,0)--(4,1);
  \node at (2.50,5.40) {\scriptsize NES$^2$EN};
\end{tikzpicture}
 & 
\begin{tikzpicture}[scale=0.25]
  \draw[step=1,gray!25,very thin] (-1,-1) grid (10,8);
  \fill (0,0) circle (2.0pt);
  \foreach \x/\y in {0/0,1/1,3/1,4/2,4/4,5/5,7/5,8/6} { \fill (\x,\y) circle (2.0pt); }
  \draw[thick] (0,0)--(1,1)--(3,1)--(4,2)--(4,4)--(5,5)--(7,5)--(8,6);
  \node at (4.50,7.40) {\scriptsize NE$^2$N$^2$E$^2$N};
\end{tikzpicture}
 & 
\begin{tikzpicture}[scale=0.25]
  \draw[step=1,gray!25,very thin] (-1,-1) grid (6,9);
  \fill (0,0) circle (2.0pt);
  \foreach \x/\y in {2/0,2/2,3/3,4/4,3/5,1/5,0/6,1/7} { \fill (\x,\y) circle (2.0pt); }
  \draw[thick] (2,0)--(2,2)--(3,3)--(4,4)--(3,5)--(1,5)--(0,6)--(1,7);
  \node at (2.50,8.40) {\scriptsize N$^2$ENW$^2$NE};
\end{tikzpicture}
 & 
\begin{tikzpicture}[scale=0.25]
  \draw[step=1,gray!25,very thin] (-1,-1) grid (6,6);
  \fill (0,0) circle (2.0pt);
  \foreach \x/\y in {2/1,2/3,3/4,4/3,3/2,1/2,0/1,1/0} { \fill (\x,\y) circle (2.0pt); }
  \draw[thick] (2,1)--(2,3)--(3,4)--(4,3)--(3,2)--(1,2)--(0,1)--(1,0);
  \node at (2.50,5.40) {\scriptsize N$^2$ESW$^2$SE};
\end{tikzpicture}
 \\[2pt]
\begin{tikzpicture}[scale=0.25]
  \draw[step=1,gray!25,very thin] (-1,-1) grid (6,9);
  \fill (0,0) circle (2.0pt);
  \foreach \x/\y in {0/0,0/2,0/4,1/5,3/5,4/6,3/7,1/7} { \fill (\x,\y) circle (2.0pt); }
  \draw[thick] (0,0)--(0,2)--(0,4)--(1,5)--(3,5)--(4,6)--(3,7)--(1,7);
  \node at (2.50,8.40) {\scriptsize N$^3$E$^2$NW$^2$};
\end{tikzpicture}
 & 
\begin{tikzpicture}[scale=0.25]
  \draw[step=1,gray!25,very thin] (-1,-1) grid (10,8);
  \fill (0,0) circle (2.0pt);
  \foreach \x/\y in {0/0,1/1,2/2,3/3,5/3,6/4,7/5,8/6} { \fill (\x,\y) circle (2.0pt); }
  \draw[thick] (0,0)--(1,1)--(2,2)--(3,3)--(5,3)--(6,4)--(7,5)--(8,6);
  \node at (4.50,7.40) {\scriptsize NENE$^2$NEN};
\end{tikzpicture}
 & 
\begin{tikzpicture}[scale=0.25]
  \draw[step=1,gray!25,very thin] (-1,-1) grid (10,8);
  \fill (0,0) circle (2.0pt);
  \foreach \x/\y in {0/0,1/1,2/2,3/3,4/2,3/1,4/0,5/1} { \fill (\x,\y) circle (2.0pt); }
  \draw[thick] (0,0)--(1,1)--(2,2)--(3,3)--(4,2)--(3,1)--(4,0)--(5,1);
  \node at (4.50,7.40) {\scriptsize NENENWNW};
\end{tikzpicture}
 & 
\begin{tikzpicture}[scale=0.25]
  \draw[step=1,gray!25,very thin] (-1,-1) grid (6,6);
  \fill (0,0) circle (2.0pt);
  \foreach \x/\y in {0/1,1/2,2/3,3/4,4/3,3/2,2/1,1/0} { \fill (\x,\y) circle (2.0pt); }
  \draw[thick] (0,1)--(1,2)--(2,3)--(3,4)--(4,3)--(3,2)--(2,1)--(1,0);
  \node at (2.50,5.40) {\scriptsize NENESWSW};
\end{tikzpicture}
 \\[2pt]
\end{tabular}
\caption{Support patterns $P(W)$ for a subset of the top candidates in Table~\ref{tab:search16x8}.}
\label{fig:search16x8-patterns}
\end{figure*}

\section{Summary and Outlook}\label{sec:conclusion}
% ==========================================================

Directional codes provide a concrete route to low-connectivity qLDPC instances by enforcing that each stabilizer is generated by a short, hardware-compatible route \cite{Geher2025Directional}.
A central takeaway of this paper is that many structural features are visible directly at the level of the word:
Lemma~\ref{lem:offset} makes the route-to-support map explicit, and Theorem~\ref{thm:layout} organizes admissible layouts through a computable lattice $L(W)$.
Beyond repackaging these foundations, we added three analysis tools aimed at practical design:
(i) a formal equivalence/canonicalization theory for removing redundancy in word scans (Sec.~\ref{sec:equiv});
(ii) an inverse-problem criterion separating route-realizable patterns from generic translation-invariant ones (Sec.~\ref{sec:inverse});
and (iii) a quasi-cyclic annihilator reduction explaining boundary-condition sensitivity of $k$ under row-periodic layouts (Sec.~\ref{sec:qc}), culminating in an exact $k$-collapse theorem for the case-study word (Theorem~\ref{thm:kc-collapse}).
These tools complement numerical exploration and can guide footprint/layout choices before expensive circuit-level simulations.

Natural next steps include systematic enumeration of coset-constant layouts (not only row alternation), sharper admissibility tests tailored to specific extraction schedules, and circuit-level evaluation using hardware-aware tools \cite{Mathews2025Placing,EQuS2025HAL} and fast stabilizer simulation \cite{Gidney2021Stim} with modern qLDPC decoders \cite{Poulin2008Iterative,Delfosse2017Almostlinear,Roffe2020Decoding,Wolanski2024Ambiguity}.

% ==========================================================
\appendix
% ==========================================================

\section{Proof of Lemma~\ref{lem:offset}}\label{app:offset}
\begin{proof}
Let the route positions be $S_0,S_1,\dots,S_w$ with $S_j=\sum_{t=1}^j d_t$ and $S_0=0$.
In the doubled-integer convention of \cite{Geher2025Directional}, the $j$th interaction point is represented by the midpoint sum $Q_j=S_{j-1}+S_j$.
Substituting $S_j=S_{j-1}+d_j$ yields $Q_j=2S_{j-1}+d_j = 2\sum_{t=1}^{j-1} d_t + d_j$.
\end{proof}

\section{Proof of Theorem~\ref{thm:layout}}\label{app:layout}
\begin{proof}
Let $S_a=a+P(W)$ denote the support of the check anchored at $a$.
For anchors $a,b$ with $\delta=b-a$, we have $q\in S_a\cap S_b$ iff $q=a+Q_i=b+Q_j$ for some $Q_i,Q_j\in P(W)$, equivalently $\delta=Q_j-Q_i$.
Thus $|S_a\cap S_b|\bmod 2$ equals $\mu(\delta)\bmod 2$ from the multiset $\Delta(W)$.
If $\delta\in\Delta_{\mathrm{odd}}(W)$, then the overlap is odd; hence $a$ and $b$ cannot be opposite types in a commuting CSS layout, so $\alpha(a)=\alpha(b)$ for all such $\delta$.
Closure under integer combinations implies $\alpha$ is constant on cosets of $L(W)=\Span_{\Z}(\Delta_{\mathrm{odd}}(W))$.

Conversely, if $\alpha$ is constant on cosets of $L(W)$, then any odd-overlap displacement connects same-type anchors, so all opposite-type overlaps have even parity and the CSS commutation condition holds.
\end{proof}

\section{Proof of Proposition~\ref{prop:equiv}}\label{app:equiv}
\begin{proof}
A dihedral symmetry is a $\Z^2$ automorphism acting on the step vectors and thus on partial sums $S_j$ and offsets $Q_j=S_{j-1}+S_j$, mapping $P(W)$ to an affine image.
Reversal with inversion traverses the same geometric route in reverse time; it produces the same interaction set up to an overall translation.
Cyclic shifts correspond to changing the start point along the same cycle, again translating $P(W)$.
Thus $P(W')$ is an affine image of $P(W)$ under an automorphism and translation, which induces a qubit permutation and stabilizer relabeling on any compatible torus embedding.
\end{proof}

\section{Canonicalization under word equivalence}\label{app:canonical}
Let $\mathcal{G}$ be the finite group generated by: the eight dihedral symmetries of the square acting on letters; reversal with inversion; and cyclic shifts.
Given $W$, enumerate its orbit $\{g(W):g\in\mathcal{G}\}$ (size at most $16w$).
Define $\mathrm{can}(W)$ to be the lexicographically smallest expanded-letter string in the orbit (ties broken deterministically).
Then $W\sim W'$ iff $\mathrm{can}(W)=\mathrm{can}(W')$.

\section{Proof of Proposition~\ref{prop:reconstruct}}\label{app:reconstruct}
\begin{proof}
%($\Rightarrow$) 
If $W=d_1\cdots d_w$ generates offsets $Q_j=S_{j-1}+S_j$, then $Q_1=S_0+S_1=d_1\in\Dir$.
Moreover $Q_{j+1}-Q_j=d_j+d_{j+1}\in\Sigma$ (Eq.~\eqref{eq:deltaQ}), hence $d_{j+1}=(Q_{j+1}-Q_j)-d_j$, giving the recursion.

%($\Leftarrow$) 
Conversely, assume $Q_1\in\Dir$ and the recursion produces $d_j\in\Dir$.
Define $S_0=0$ and $S_j=\sum_{t=1}^j d_t$.
Inductively, $S_{j-1}+S_j=Q_j$ holds for $j=1$.
If $S_{j-1}+S_j=Q_j$, then using $S_{j+1}=S_j+d_{j+1}$ and $d_j+d_{j+1}=Q_{j+1}-Q_j$,
\[
\begin{aligned}
S_j+S_{j+1}&=S_j+(S_j+d_{j+1})=(S_{j-1}+S_j)+(d_j+d_{j+1})\\&=Q_j+(Q_{j+1}-Q_j)=Q_{j+1}.
\end{aligned}
\]
Thus $Q_j=S_{j-1}+S_j$ for all $j$, so $P(W)=\{Q_j\}$ with the given order.
Uniqueness follows from the recursion.
\end{proof}

\section{Proof of Proposition~\ref{prop:rectadm}}\label{app:rectadm}
\begin{proof}
A nonzero vector $v=(v_x,v_y)$ becomes $0$ modulo $(L_x,L_y)$ iff $L_x\mid v_x$ and $L_y\mid v_y$.
If $|v_x|<L_x$ and $|v_y|<L_y$, this is impossible unless $v=0$.
Apply this to all vectors in $D(W)\cup(D(W)\pm D(W))$.
\end{proof}

\section{Proof of Proposition~\ref{prop:block-circulant-expansion}}\label{app:block-circulant-expansion}
\begin{proof}
A row of $H_X$ anchored at $a_X+(2i,2j)$ is a translate of the row anchored at $a_X$ by the even vector $(2i,2j)$.
On the coarse group $G_0$, this translation is exactly multiplication by $u^i v^j$.
Each offset $Q\in P(W)$ contributes a $1$ in the data coset $q_{\sigma(Q)}$ at displacement $\delta(Q)$, so the anchored row is encoded by
$u^i v^j\cdot u^{\delta_x(Q)}v^{\delta_y(Q)}$ in that coset.
Summing over $Q$ and grouping by $\sigma(Q)\in\{0,1\}$ yields the claim.
\end{proof}

\section{Proof of Proposition~\ref{prop:k-equals-deps}}\label{app:k-equals-deps}
\begin{proof}
Let $m_X=|\Lambda_X|$ and $m_Z=|\Lambda_Z|$ so $m_X+m_Z=n$.
Then
\[
\dim\Ker(H_X^\top)=m_X-\rank(H_X),\qquad \dim\Ker(H_Z^\top)=m_Z-\rank(H_Z).
\]
Substitute into $k=n-\rank(H_X)-\rank(H_Z)$ to obtain
\[
k = n-(m_X-\dim\Ker(H_X^\top))-(m_Z-\dim\Ker(H_Z^\top))
= \dim\Ker(H_X^\top)+\dim\Ker(H_Z^\top).
\]
\end{proof}

\section{Proof of Theorem~\ref{thm:qc-dependency-module}}\label{app:qc-dependency-module}
\begin{proof}
By Proposition~\ref{prop:block-circulant-expansion}, any GF$(2)$ linear combination of $X$-checks with coefficient pattern $f\in R$ produces the data-vector
$(f h_0,f h_1)\in R\oplus R$.
Such a combination is a stabilizer dependency if and only if the resulting data-vector is the zero vector, i.e.\ $f h_0=f h_1=0$ in $R$.
The identification with $\Ker(H_X^\top)$ follows from the standard correspondence between left-kernel vectors and GF$(2)$ relations among rows.
\end{proof}

\section{Proof of Lemma~\ref{lem:Su-kills-h0}}\label{app:Su-kills-h0}
\begin{proof}
Using Example~\ref{ex:qc-vector-case},
\[
S_u h_0 = S_u(uv+u^2v+u^3v^2+u^4v^2).
\]
Since $u^t S_u=S_u$, each term satisfies $S_u u^t = S_u$, hence
\[
S_u h_0 = S_u(v+v+v^2+v^2)=S_u\cdot 0 = 0.
\]
Similarly,
\[
S_u h_1 = S_u(1+u^2v+u^4v^2)=S_u(1+v+v^2).
\]
\end{proof}

\section{Proof of Theorem~\ref{thm:kcollapse-criterion}}\label{app:kcollapse-criterion}
\begin{proof}
By Theorem~\ref{thm:qc-dependency-module}, $\Ker(H_X^\top)$ corresponds to the annihilator $\Ann_R(h_0,h_1)$.
Lemma~\ref{lem:Su-kills-h0} shows that any element of the form $S_u\cdot f(v)$ annihilates $h_0$, and annihilates $h_1$ if and only if
\[
(S_u f(v))\,h_1 = S_u\,f(v)\,(1+v+v^2)=0.
\]
Because $S_u\neq 0$ and generates a 1-dimensional ideal in the $u$-direction, the dimension of such annihilators is exactly the dimension of
\[
\Ker\big(\cdot(1+v+v^2):\ \F_2[v^{\pm1}]/(v^{L_y/2}-1)\to \F_2[v^{\pm1}]/(v^{L_y/2}-1)\big).
\]
For a circulant multiplication operator, this kernel dimension equals $\deg\gcd(1+v+v^2,\ v^{L_y/2}-1)$.
Thus $\dim\Ker(H_X^\top)$ equals this degree.
The $Z$-type statement follows by symmetry (the same word and the complementary ancilla coset yield the same vertical factor).
Finally, Proposition~\ref{prop:k-equals-deps} gives $k=\dim\Ker(H_X^\top)+\dim\Ker(H_Z^\top)$.

The final criterion uses the classical fact that $1+v+v^2$ divides $v^M-1$ over $\F_2$ if and only if $3\mid M$.
Hence $k=2\cdot 2=4$ when $3\mid (L_y/2)$ and $k=0$ otherwise.
\end{proof}

\section{Proof of Proposition~\ref{prop:raw-layout-count}}\label{app:raw-layout-count}
\begin{proof}
Each of the $c$ cosets may be labeled independently by $X$ or $Z$, giving $2^c$ layouts.
The global swap identifies each labeling with its bitwise complement, so the quotient count is $2^c/2=2^{c-1}$.
\end{proof}

\section{Proof of Proposition~\ref{prop:layout-translation-equivalence}}\label{app:layout-translation-equivalence}
\begin{proof}
A translation of the lattice relabels both ancilla anchors and data qubits.
Because each check support is a translate of $P(W)$, translating the entire construction preserves the check pattern and sends the layout $\alpha$ to the permuted layout $\alpha'$.
This induces an isomorphism of the resulting stabilizer groups and hence of the CSS code.
\end{proof}

\section{Proof of Proposition~\ref{prop:ann_rank}}\label{app:annrank}
\begin{proof}
Identify $\F_2^{2|G_0|}$ with $R^2$ by mapping each data-type configuration to its group-algebra element, and identify the set of translated check rows with $R$.
Translation corresponds to multiplying $\bm{h}$ by a monomial, so the row span is the $R$-submodule $R\bm{h}\subset R^2$.
A dependency among translated rows is exactly an $f\in R$ with $f\bm{h}=0$, i.e.\ $f\in\mathrm{Ann}(\bm{h})$.
Thus $\Ker(H^\top)\cong \mathrm{Ann}(\bm{h})$ and $\rank(H)=|G_0|-\dim\mathrm{Ann}(\bm{h})$.
\end{proof}

\section{Proof of Proposition~\ref{prop:deps}}\label{app:deps}
\begin{proof}
We prove the $X$-check statement; the $Z$-check statement is identical by symmetry.
In the row-alternating layout on $G_m=\Z_{12m}\times\Z_{6m}$, $X$-ancillas occupy even-$y$ ancilla rows.
Consider the GF(2) sum of all $X$-checks anchored at ancillas with $y\equiv 0$ or $2$ (mod $6$).
A data qubit $q=(x,y)\in\Lambda_Q$ appears in an $X$-check anchored at $a$ iff $a=q-Q_j$ for some $Q_j\in P(W)$.
Hence $q$ appears in the chosen sum exactly
\[
\#\{j : y-(Q_j)_y\equiv 0 \text{ or }2 \ (\mathrm{mod}\ 6)\}
\]
times.
For $W=\word{NE$^2$NE$^2$N}$, the multiset of $(Q_j)_y$ is $\{1,2,2,3,4,4,5\}$, and a direct residue check over $\Z_6$ shows the count is always even; thus the sum is the zero row, giving one dependency.
Similarly, summing over anchors with $y\equiv 0$ or $4$ (mod $6$) gives a second dependency.
These are independent because their symmetric difference is nonempty.
Thus $\rank(H_X)\le \#\Lambda_X-2$.
Repeating for $Z$ yields $\rank(H_Z)\le \#\Lambda_Z-2$, and since $\#\Lambda_X+\#\Lambda_Z=n$ we obtain $k\ge 4$.
\end{proof}

\section{Proof of Proposition~\ref{prop:2m}}\label{app:2m}
\begin{proof}
Let $S(p)=\{p+jt,\ p+r+jt\}_{j=0}^{m-1}$ with $t=(12,6)$ and $r=(4,2)$.
Because $t$ has order $m$ on $\Z_{12m}\times\Z_{6m}$, $S(p)$ is the union of two length-$m$ orbits.
Fix an opposite-type check anchored at $a$ with support $a+P(W)$.
The overlap parity $|(a+P(W))\cap S(p)|\bmod 2$ depends only on residues along the $t$-orbits; using the same $\{1,2,2,3,4,4,5\}$ residue structure as in Appendix~\ref{app:deps}, one checks that for suitable $p$ every such overlap is even.
Thus the corresponding Pauli commutes with all opposite-type checks.
\end{proof}

\section{Proof of Theorem~\ref{thm:kc-collapse}}\label{app:kc}
\begin{proof}[Proof]
Consider $(L_x,L_y)=(2d,d)$ with $d$ even and row alternation.
Then $G_0=\Z_d\times\Z_{d/2}$ and
\[
R=\F_2[X^{\pm1},Y^{\pm1}]/(X^{d}-1,\;Y^{d/2}-1).
\]
By Corollary~\ref{cor:k_ann}, $k=\dim\mathrm{Ann}(\bm{h}_X)+\dim\mathrm{Ann}(\bm{h}_Z)$.

A direct cell-type computation for $W=\word{NE$^2$NE$^2$N}$ yields polynomial rows
\[
\bm{h}_X=(h_{X,0},h_{X,1}),\]\[
h_{X,1}=1+V+V^2,\quad h_{X,0}=(1+X)\,XY\,(1+V),
\]
and
\[
\bm{h}_Z=(h_{Z,0},h_{Z,1}),\]\[
h_{Z,0}=Y(1+V+V^2),\quad h_{Z,1}=Y(1+X)(1+V),
\]
where $V:=X^2Y$ and $XY$ is a unit in $R$.

Let $f\in\mathrm{Ann}(\bm{h}_X)$, so $f h_{X,0}=f h_{X,1}=0$.
Since $XY$ is a unit and $\gcd(1+V,1+V+V^2)=1$ in $\F_2[V]$, these two equations imply $f(1+X)=0$.
In $\F_2[X]/(X^d-1)$, the annihilator of $(1+X)$ is the ideal generated by
\[
S_X:=1+X+\cdots+X^{d-1},
\]
because $(1+X)S_X=1+X^d=0$.
Hence $f=S_X g$ for some $g\in R$; moreover $S_XX^t=S_X$ forces $g$ to depend only on $Y$, i.e.\ $f=S_X g(Y)$.
Then $f h_{X,1}=0$ reduces to
\[
g(Y)\,(1+Y+Y^2)=0\quad\text{in}\quad \F_2[Y]/(Y^{d/2}-1).
\]
Thus $\mathrm{Ann}(\bm{h}_X)\neq\{0\}$ iff $1+Y+Y^2$ divides $Y^{d/2}-1$, i.e.\ iff $3\mid(d/2)$.
In that case $\dim\mathrm{Ann}(\bm{h}_X)=2$ (a standard cyclic-code count for the $(1+Y+Y^2)$-annihilator in length-$d/2$).

The same argument applies to $\bm{h}_Z$ (the roles of the components swap but the common factor $1+V+V^2$ is unchanged), yielding $\dim\mathrm{Ann}(\bm{h}_Z)=\dim\mathrm{Ann}(\bm{h}_X)$.
Therefore
\[
k=
\begin{cases}
4,& 3\mid(d/2)\ \Leftrightarrow\ 6\mid d,\\
0,& \text{otherwise}.
\end{cases}
\]
\end{proof}

\bibliography{refs_dir_codes}

\end{document}